\documentclass[11pt]{article}
\usepackage[margin=1in]{geometry}
\usepackage{amsmath,amssymb,amsthm,mathtools,microtype,enumitem}
\usepackage[T1]{fontenc}
\usepackage{lmodern}
\usepackage{booktabs,tabularx}
\usepackage{needspace}
\usepackage{xcolor}
\usepackage[colorlinks=true,linkcolor=blue!60!black,citecolor=blue!60!black,urlcolor=blue!60!black]{hyperref}
\usepackage[capitalize,noabbrev]{cleveref}
\hypersetup{pdftitle={Sampling Matchings in Near-linear Time},
pdfauthor={Tianshun Miao and Yitong Yin},
pdfsubject={Glauber mixing, work-efficient parallel sampling, and approximate counting}}
\newtheorem{theorem}{Theorem}[section]
\newtheorem{lemma}[theorem]{Lemma}
\newtheorem{proposition}[theorem]{Proposition}
\newtheorem{corollary}[theorem]{Corollary}
\theoremstyle{definition}
\theoremstyle{remark}
\theoremstyle{definition}\newtheorem{introdefinition}{Definition}[section]
\theoremstyle{remark}
\crefname{introdefinition}{Definition}{Definitions}
\Crefname{introdefinition}{Definition}{Definitions}
\crefname{introremark}{Remark}{Remarks}
\Crefname{introremark}{Remark}{Remarks}
\theoremstyle{plain}
\newcommand{\wtO}{\widetilde O}
\newcommand{\tmix}{t_{\mathrm{mix}}}
\newcommand{\calM}{\mathcal M}
\newcommand{\E}{\mathbb E}
\newcommand{\Prb}{\mathbb P}
\newcommand{\Var}{\operatorname{Var}}
\newcommand{\Cov}{\operatorname{Cov}}
\newcommand{\Ent}{\operatorname{Ent}}
\newcommand{\Diag}{\operatorname{Diag}}

\newcommand{\cE}{\mathcal E}
\newcommand{\cL}{\mathcal L}
\newcommand{\cV}{\mathcal V}
\newcommand{\one}{\mathbf 1}
\newcommand{\rise}[2]{(#1)^{\overline{#2}}}
\newcommand{\TV}{d_{\mathrm{TV}}}
\DeclareMathOperator{\gap}{gap}
\setlist{itemsep=3pt,topsep=5pt}
\title{Sampling Matchings in Near-linear Time}
\author{Tianshun Miao\qquad Yitong Yin\\[0.6em]
\normalsize State Key Laboratory for Novel Software Technology,\\
\normalsize New Cornerstone Science Laboratory,\\
\normalsize Nanjing University, China\\[0.3em]
\normalsize\texttt{miaotianshun@smail.nju.edu.cn, yinyt@nju.edu.cn}}
\date{}
\begin{document}
\hypersetup{pageanchor=false}
\pagenumbering{gobble}
\maketitle

\begin{abstract}

%
For every fixed activity $\lambda>0$, we establish three results for the
monomer--dimer model on an $n$-vertex simple graph $G$ with $m\ge1$ edges and maximum degree $\Delta$.
\begin{enumerate}[label=(\roman*),leftmargin=*,itemsep=4pt,topsep=5pt]
\item \textbf{Near-linear mixing and sampling.}
Single-edge Glauber dynamics has mixing time
$O_\lambda(m[\log^2 n+\log(1/\varepsilon)])$, giving a near-linear-time approximate sampler.
\item \textbf{Work-efficient parallel sampling.}
We simulate the same Glauber dynamics in parallel using
$\wtO_\lambda(m+n)$ work and $\wtO_\lambda(\min\{\Delta,m^{1/3},\sqrt n\})$
depth with high probability.
\item \textbf{Fast approximate counting.}
We estimate the partition function within relative error $\varepsilon$
in $\wtO_\lambda(n^2/\varepsilon^2)$ work.
For dense graphs with $m=\Theta(n^2)$, this is near-linear in the input
size.
\end{enumerate}
For the mixing theorem, we establish a general log--Sobolev criterion
based on field-dynamics spectral stability, with only logarithmic
dependence on the inverse occupied-marginal lower bound.
Parallelism uses a matching-specific analysis of occupation-interval dependencies.
Counting uses monomer-preconditioned Jerrum--Sinclair dynamics, whose
parameters are learned efficiently by Glauber dynamics.
\end{abstract}
\clearpage
\begingroup\small\setlength{\parskip}{0pt}
\tableofcontents
\endgroup
\clearpage
\pagenumbering{arabic}
\hypersetup{pageanchor=true}

\section{Introduction}\label{sec:intro}

Let $G=(V,E)$ be a simple graph with $n$ vertices and $m$ edges.
The monomer--dimer distribution at activity $\lambda>0$ and its partition function  are defined as
\[
 \mu_{G,\lambda}(M)=\frac{\lambda^{|M|}}{Z_G(\lambda)},
 \qquad Z_G(\lambda)=\sum_{M\in\calM(G)}\lambda^{|M|},
\]
where $\calM(G)$ is the set of all matchings of $G$.
At $\lambda=1$, this is the uniform distribution, and $Z_G(1)$ is the number
of matchings.

Matchings have played a foundational role in the study of sampling and
approximate counting. The seminal work of Jerrum and Sinclair
\cite{JerrumSinclair} established polynomial-time approximation for the
monomer--dimer partition function through a rapidly mixing Markov chain.
The conductance and canonical-path methods developed in this line of
work, and subsequent refinements \cite{Sinclair92}, helped shape the
modern theory of Markov chain Monte Carlo algorithms. Together with
general reductions between sampling and counting \cite{JVV86}, these
results settled polynomial-time tractability.

\paragraph{Near-linear mixing and sampling.}
A natural local chain for sampling matchings is single-edge Glauber
dynamics (GD). Assuming $m\ge 1$, each step chooses an edge uniformly,
removes it if present, and then inserts it with probability
$\lambda/(1+\lambda)$ if both endpoints are unmatched; otherwise it
remains absent. Write $P_{\mathrm{GD}}$ for this transition matrix, whose
stationary distribution is $\mu_{G,\lambda}$.

For an ergodic Markov chain  $P$ on $\Omega$ with stationary distribution $\mu$, its \emph{mixing time} is defined by
\begin{align*}
\tmix(P,\varepsilon)
 &=\min\left\{t\in\mathbb Z_{\ge0}:
       \max_{x\in\Omega}d_{\mathrm{TV}}(P^t(x,\cdot),\mu)
       \le\varepsilon\right\},\qquad  0<\varepsilon<1.
\end{align*}
Here $d_{\mathrm{TV}}(\rho,\mu)=\frac12\sum_{x\in\Omega}|\rho(x)-\mu(x)|$ denotes the 
\emph{total variation distance} between probability distributions $\rho,\mu$ on a finite set $\Omega$.
For matching GD, write
 $\tmix(\varepsilon)=\tmix(P_{\mathrm{GD}},\varepsilon)$.

Using spectral independence, Chen, Liu, and Vigoda \cite{CLV21}
established $\tmix(1/4)=O(m\log n)$ at fixed activity and constant
maximum degree.

The unbounded-degree regime is more subtle. Although the monomer--dimer
model has no finite positive-activity phase transition \cite{HL72},
the rate of correlation decay deteriorates as the degree grows
\cite{BGKNT07}, already on high-degree trees. This loss of
degree-uniform correlation decay gives the regime a
\emph{near-critical} character. For comparison, recent work on the
hard-core and Ising models at their uniqueness thresholds proves
superlinear polynomial lower bounds on worst-case Glauber mixing
\cite{CCYZ25}. These developments motivate a fundamental question:
\emph{Can Glauber dynamics for matchings achieve near-linear mixing
without any degree restriction?}

We answer this question affirmatively: ordinary single-edge Glauber
dynamics achieves near-linear mixing and sampling time on every simple graph.

\begin{theorem}[Near-linear mixing and sampling]
\label{thm:intro-sequential}
For every activity $\lambda>0$, every simple graph $G$ with $n$ vertices
and $m\ge1$ edges, and $0<\varepsilon\le1/2$, single-edge Glauber dynamics
for $\mu_{G,\lambda}$ satisfies
\[
 \tmix(\varepsilon)
 =O_\lambda\!\left(m[\log^2 n+\log(1/\varepsilon)]\right).
\]
Consequently, a random matching with distribution $\widehat\mu$ satisfying
$d_{\mathrm{TV}}(\widehat\mu,\mu_{G,\lambda})\le\varepsilon$
can be generated in
$O_\lambda(n+m[\log^2 n+\log(1/\varepsilon)])$ time.
\end{theorem}

By storing each vertex’s matched neighbor (or a marker indicating that it is unmatched), we can check whether an edge can be inserted and update the matching in constant time. Thus each Glauber update takes constant time, yielding the stated running-time bound. 

\paragraph{Work-efficient parallel sampling.}
Parallel computation asks how much sequential dependence is inherent
in sampling. The work of Mulmuley, Vazirani, and Vazirani \cite{MVV87}
made finding a maximum matching a landmark example of randomized
polylogarithmic-depth computation. Sampling and approximate counting
pose different challenges. Teng \cite{Teng95} identified an obstruction
to NC simulation of the Broder--Jerrum--Sinclair approach to bipartite
perfect matchings. This is a barrier to simulating prescribed chain
trajectories, not a lower bound for sampling itself.

For the monomer--dimer model, recent degree-dependent mixing bounds
and parallel simulation techniques yield polylogarithmic depth at
fixed activity when the maximum degree is polylogarithmic
\cite{CFJMYZ25,LY25,HKLYZ26}. The challenge is to handle arbitrary
degrees:
\emph{Can matchings be sampled with substantially sublinear depth
while preserving polynomial or even near-linear total work?}

We achieve both guarantees by simulating the same single-edge
Glauber trajectory in parallel.

\begin{theorem}[Work-efficient parallel sampling]
\label{thm:intro-parallel}
There is a randomized CREW PRAM algorithm that, given a simple graph $G$
with $n$ vertices, $m\ge1$ edges, and maximum degree $\Delta$,
an activity $\lambda>0$, and $0<\varepsilon\le1/2$, returns a random matching with distribution
$\widehat\mu$ satisfying
$d_{\mathrm{TV}}(\widehat\mu,\mu_{G,\lambda})\le\varepsilon$.
The algorithm uses $\wtO_\lambda(m+n)$ work on every run and with high probability has depth
\[
 \wtO_\lambda\!\left(\min\{\Delta,m^{1/3},\sqrt n\}\right).
\]
\end{theorem}

Here $\wtO_\lambda$ suppresses polylogarithmic factors in graph size and $1/\varepsilon$. 
The total-variation guarantee is unconditional.
At fixed activity and inverse-polynomial accuracy and failure
probability, the depth is polylogarithmic when the maximum degree
is polylogarithmic, and at most $\wtO_\lambda(\sqrt n)$ without
any degree restriction.

\paragraph{Fast approximate counting.}
At fixed activity, standard annealing reductions estimate the partition
function $Z_G(\lambda)$ within relative error $\varepsilon$ using
$\wtO(n/\varepsilon^2)$ samples at intermediate activities
\cite{SVV09,HKLYZ26}. Combined with our sampler, these reductions give
an independent-sample annealing baseline of $\wtO_\lambda(mn/\varepsilon^2)$
work, after one-time $O(m+n)$ input preprocessing.

For vertex-spin systems on bounded-degree graphs with near-linear-time
samplers, the analogous baseline is $\wtO(n^2/\varepsilon^2)$:
each sample costs $\wtO(n)$ rather than the edge-based $\wtO(m)$
for matchings. Recent advances improve on this baseline for several
subcritical spin systems \cite{AFFGW25,CCLZ26}. For matchings,
the loss of degree-uniform correlation decay gives the unbounded-degree
regime a near-critical character. This raises a further question:
\emph{Can approximate counting beat the independent-sample annealing baseline
without any degree restriction?}

We give an approximate counting algorithm with
$\wtO_\lambda(n^2/\varepsilon^2)$ work on every simple graph,
achieving near-linear work in the input size on dense graphs at constant accuracy.

\begin{theorem}[Fast approximate counting]
\label{thm:intro-counting}
There is a randomized algorithm that, given a simple graph $G$ with
$n$ vertices, an activity $\lambda>0$, and $0<\varepsilon\le1/2$,
returns $\widehat Z$ satisfying
\[
 \Pr\!\left((1-\varepsilon)Z_G(\lambda)\le\widehat Z
                     \le(1+\varepsilon)Z_G(\lambda)\right)\ge\frac34.
\]
The algorithm uses $\wtO_\lambda(n^2/\varepsilon^2)$ work on every run.
\end{theorem}

For dense graphs with $m=\Theta(n^2)$, the bound becomes
$\wtO_\lambda(m/\varepsilon^2)$: near-linear in the input size
at a constant accuracy. This removes a factor $n$
from the independent-sample annealing baseline.

\subsection{A general mixing criterion from field-dynamics spectral stability}
\label{sec:mixing-criterion}

The near-linear Glauber mixing bound follows from a general criterion
that applies beyond matchings. Its assumptions combine field-dynamics
spectral stability with two marginal bounds.

Let $\mu$ be positive on a downward-closed family
$\Omega\subseteq2^{[N]}$ containing at least two states.
For $\tau\in\Omega$, a \emph{positive pinning} conditions on
$\tau\subseteq S$ and leaves the other coordinates unpinned.
We write $\mu^\tau$ for this conditional distribution, viewed on
the unpinned coordinates $[N]\setminus\tau$; explicitly,
\[
 \mu^\tau(T)
 =\frac{\mu(\tau\cup T)}
        {\sum_{U\in\Omega:\,\tau\subseteq U}\mu(U)},
 \qquad T\in\Omega^\tau,
\]
where
$\Omega^\tau=\{T\subseteq[N]\setminus\tau:\tau\cup T\in\Omega\}$.
Define the decreasing scalar-field laws
\[
 \nu_t(S)=\frac{(1-t)^{|S|}\mu(S)}
                   {\sum_{T\in\Omega}(1-t)^{|T|}\mu(T)},
 \qquad 0\le t<1.
\]
We use \emph{field-dynamics spectral stability} as shorthand for spectral
stability with respect to field dynamics, in the formulation of
Chen, Chen, Yin, and Zhang \cite{CCYZ25}, building on the localization
framework of Chen and Eldan \cite{CE22}. We use the equivalent covariance
characterization in \cite[Proposition~3.3]{CCCYZ25}.

\begin{introdefinition}[Field-dynamics spectral stability]\label{def:intro-stability}
The law $\mu$ satisfies field-dynamics spectral stability with rate
$C:[0,1)\to[0,\infty)$ if, for every $0\le t<1$ and every
positive pinning $\tau\in\Omega$,
\begin{equation}\label{intro-eq:stability}
 \Cov_{\nu_t^\tau}(X)\preceq C(t)D_{\nu_t^\tau},
 \qquad D_\rho:=\Diag(\E_\rho X).
\end{equation}
Here $X$ is the residual indicator vector. Coordinates with zero marginal
can be omitted. Equivalently, the largest eigenvalue of
$D_\rho^{-1/2}\Cov_\rho(X)D_\rho^{-1/2}$ is at most $C(t)$ for
$\rho=\nu_t^\tau$; an empty matrix satisfies the condition vacuously.
\end{introdefinition}

This is a property of the negative-field localization associated with the
field dynamics introduced by Chen, Feng, Yin, and Zhang \cite{CFYZ21}. To see the
connection, sample $S\sim\mu$ and reveal each occupied coordinate at an
independent uniform time in $[0,1]$. Writing $Y_t$ for the revealed set,
each positive-probability observation $Y_t=\tau$ has residual posterior
$\nu_t^\tau$. Spectral stability
bounds the instantaneous conditional variance rate of the posterior expectation
by $C(t)/(1-t)$ times the posterior variance
\cite[Appendix~A.1]{CCYZ25}.
The normalization uses marginal \emph{means}, not variances; it is distinct
from the usual spectral-independence normalization
\cite[Remark 3.4]{CCCYZ25}.

To state the analytic conclusion, for an irreducible $\mu$-reversible Markov kernel $P$,
write
\begin{align*}
 \cE_P(f,g)&=\frac12\sum_{S,T\in\Omega}\mu(S)P(S,T)
                  (f(S)-f(T))(g(S)-g(T)),\\
 \Ent_\mu(h)&=\E_\mu[h\log h]-\E_\mu[h]\log\E_\mu[h],
\end{align*}
for real $f,g$ and $h\ge0$, with $0\log0=0$.
These are the \emph{Dirichlet form} and \emph{entropy}, respectively.
The \emph{log--Sobolev inequality} (LSI) takes the form
\[
 \Ent_\mu(f^2)\le\alpha_{\rm LS}^{-1}\cE_P(f,f),
 \qquad f:\Omega\to\mathbb R.
\]
It controls entropy by the average squared change in one update.
The constant $\alpha_{\rm LS}$ and the other functional-inequality
constants are formally defined in Section~\ref{sec:functional-inequalities}.
For Glauber dynamics, $P_{\mathrm{GD}}$ chooses a coordinate uniformly and
resamples it conditional on the others; write $\cE_{\mathrm{GD}}=\cE_{P_{\mathrm{GD}}}$.

Under the assumptions below, we obtain a log--Sobolev inequality with only
logarithmic dependence on the inverse occupied-marginal lower bound.

\begin{theorem}[Glauber mixing from field-dynamics spectral stability]\label{thm:intro-boolean}\label{abs-thm:mix}\label{green-thm:strengthening}
Let $N\ge1$ be an integer, and let $\mu$ be a probability law with downward-closed
support $\Omega\subseteq2^{[N]}$ containing at least two states.
Let $\kappa\ge1$, $0<b\le1$, and let
$C:[0,1)\to[1,\infty)$ be nonincreasing. Suppose that:
\begin{enumerate}[label=(\roman*)]
\item \emph{Field-dynamics spectral stability.}
The law $\mu$ satisfies Definition~\ref{def:intro-stability} with rate $C$, and
\begin{equation}\label{eq:integrated-stability}
 I:=\int_0^1\frac{C(t)-1}{1-t}\,dt<\infty.
\end{equation}
\item \emph{Bounded marginal ratios.} For every $S\in\Omega$ and
$i\in[N]\setminus S$ with $S\cup\{i\}\in\Omega$,
\begin{equation}\label{intro-eq:kappa}
 \frac{\mu(S\cup\{i\})}{\mu(S)}\le\kappa.
\end{equation}
\item \emph{Occupied-marginal lower bounds under positive pinning.} For every
$\tau\in\Omega$ and $i\in[N]\setminus\tau$ with
$\tau\cup\{i\}\in\Omega$,
\begin{equation}\label{intro-eq:b}
 \mu^\tau(i\in S)\ge b.
\end{equation}
\end{enumerate}
Then Glauber dynamics with stationary law $\mu$ 
satisfies the log--Sobolev inequality
\[
 \Ent_\mu(f^2)\lesssim
 \kappa \mathrm{e}^I N\log\frac{N\kappa}{b}\,\mathcal E_{\mathrm{GD}}(f,f),
 \qquad f:\Omega\to\mathbb R.
\]
Write $\mu_{\min}=\min_{S\in\Omega}\mu(S)$. For every $0<\varepsilon\le1/2$, the mixing-time bound is
\begin{equation}\label{abs-eq:mix}
 \tmix(\varepsilon)\lesssim
 \kappa \mathrm{e}^I N
 \left[\log\frac{N\kappa}{b}\,
       \log\log\frac{\mathrm{e}^2}{\mu_{\min}}+\log(1/\varepsilon)\right].
\end{equation}
In particular, for $N\ge2$, the bounds $\kappa,I=O(1)$ and $b\ge N^{-O(1)}$
imply an inverse LSI constant $O(N\log N)$ and mixing time
$O(N[\log^2 N+\log(1/\varepsilon)])$.
\end{theorem}

Theorem~\ref{thm:intro-boolean} upgrades the scalar-field spectral-gap
criterion of \cite[Theorem 1.16 and Lemma 1.18]{CCCYZ25} to LSI.
Compared with the SI-based entropy-factorization bounds of
\cite{CLV21,BCCPSV22}, our log--Sobolev bound has no explicit
maximum-degree factor and depends only logarithmically on $1/b$.
Thus polynomially small occupied marginals incur only a logarithmic loss.
The marginal hypotheses also differ.
Unlike full marginal boundedness, which controls every feasible spin
under arbitrary pinning, our assumptions bound marginal ratios above
under arbitrary pinning in (ii), and occupied marginals below only
under positive pinning in (iii).

Field dynamics already gives degree-free spectral-gap and MLSI
comparisons \cite{CFYZ21,CFYZ22}, but the latter criterion requires
relative marginal stability under further pinning, which we do not
assume. Entropic-independence criteria give MLSI for down--up walks
without marginal lower bounds, and restricted entropy factorization
under coordinatewise field control \cite{AJKPV22}.
Our criterion instead uses integrable stability along scalar
decreasing fields to obtain a full LSI for ordinary GD.
These are alternative sufficient conditions, not uniformly ordered;
Section~\ref{sec:related-work} gives further comparisons.

For matchings with $m\ge1$ edges, the substitution is particularly simple:
we have $N=m$, $\kappa=\max\{1,\lambda\}$, $b=\frac{\lambda}{(1+\lambda\Delta)^2}$, where $\Delta$ is the maximum degree.
The degree-free covariance bound
\cite{ConcurrentHolant26}, combined with the low-activity bound in
\cite[Theorem 1.14 and Section 5.3]{CCCYZ25}, implies an integrable
spectral-stability rate with $I=O_\lambda(1)$; see
\cref{prop:holant-edge}.
Meanwhile, $\Delta<n$ enters only through $\log(1/b)=O_\lambda(\log n)$.
This is how the theorem avoids a multiplicative maximum-degree loss.
This gives an inverse LSI constant
$O_\lambda(m\log n)$ for single-edge GD; see Corollary~\ref{thm:main}.

\subsection{Technical overview}\label{sec:proof-ideas}

\paragraph{Near-linear Glauber mixing.}
The analytic step is to strengthen spectral-gap control to LSI through
low-degree approximation. Let $\cV_j$ be the space of polynomials of
degree at most $j$ in the occupation indicators, and let $\Pi_j$ be its
$L^2(\mu)$-orthogonal projection. We combine two estimates:
\[
 \|f\|_{2,\mu}^2
 =\|f-\Pi_jf\|_{2,\mu}^2+\|\Pi_jf\|_{2,\mu}^2
 \le \underbrace{\frac{A}{j+1}\mathcal E_{\mathrm{GD}}(f,f)}_{
          \text{Dirichlet-form bound}}
     +\underbrace{e^{\Lambda j}\|f\|_{1,\mu}^2}_{
          \text{low-degree norm bound}}.
\]
Choosing the cutoff according to $t$, or using the full space when
that cutoff is too large, gives the \emph{super-Poincar\'e inequality}
\[
 \|f\|_{2,\mu}^2\le t\,\mathcal E_{\mathrm{GD}}(f,f)
             +\exp(K/t)\|f\|_{1,\mu}^2,
 \qquad t>0,\qquad K=A\Lambda.
\]
This is the classical super-Poincar\'e route to LSI
\cite{Wang00SPI,GM15}, related to Nash methods \cite{DSC96Nash}.
The exponential profile gives a defective LSI; centering and the
$j=0$ Poincar\'e bound remove the defect \cite{Wang00SPI,Rothaus}.
Proposition~\ref{abs-prop:degree-to-functional} gives a self-contained
rare-support and truncation proof, with inverse LSI constant $O(A\Lambda)$.
Our contribution is to derive both approximation estimates from
field-dynamics spectral stability and marginal bounds, without
requiring GD to preserve the polynomial spaces.

For the Dirichlet-form bound, we interpolate auxiliary insertion--deletion
forms along the down walk of field-dynamics localization
\cite{CFYZ21,CE22,CCYZ25}. A Bochner-type interpolation
(cf.~\cite[Lemmas 2.2--2.3]{Bochner}) tracks the best approximation constant:
optimality cancels derivative terms, and spectral stability of the
posteriors bounds its deterioration. At vanishing retention, normalized
monomials become orthogonal, with Dirichlet-form values approaching
their degrees. Integrating from this limit yields $A=N(1+\kappa)e^I$
(Lemma~\ref{abs-thm:hierarchy}).

The second estimate follows from
$\|g\|_\infty\le e^{\Lambda j/2}\|g\|_2$ for $g\in\cV_j$,
by projection duality. For feasible $S$ and $X\sim\mu$, testing it on
$\one_{\{S\subseteq X\}}$ shows why a marginal hypothesis is needed:
it requires $\Prb(S\subseteq X)\ge e^{-\Lambda|S|}$.
Our occupied-marginal lower bounds under positive pinning ensure
$\Prb(S\subseteq X)\ge b^{|S|}$ and, crucially, remain available after
each further positive pinning. Together with bounded marginal ratios,
they control forward differences, which lower polynomial degree.
Induction gives the norm bound with
$\Lambda=2\log(1+N\sqrt{(1+\kappa)/b})$
(Lemma~\ref{abs-lem:christoffel-rho-b}). Thus $1/b$ enters the LSI
coefficient only logarithmically. When $\kappa,I=O(1)$ and
$b\ge N^{-O(1)}$, we obtain inverse LSI constant $O(N\log N)$;
the LSI startup estimate and the spectral-gap accuracy estimate then
give the mixing bound of Theorem~\ref{thm:intro-boolean}.

\paragraph{Work-efficient parallel simulation.}
We represent each proposed occupation in a fixed GD update table by an
interval ending at that edge's next update. Overlapping intervals on
incident edges conflict, and greedy parallel peeling \cite{BFS12}
reproduces the trajectory. Rejected proposals can transmit dependencies,
so a long execution yields a witness alternating between accepted
occupations and rejected proposals connecting them
(see \eqref{eq:interval-witness}).

The conflict graph and priorities share the same randomness.
Auxiliary no-op padding controls rejection probabilities after the
accepted trajectory is revealed. Counting the witnesses uses the
matching constraint: two disjoint accepted edges have at most four
possible connecting edges. The resulting bounds give
$\wtO_\lambda(\min\{\Delta,\sqrt n,m^{1/3}\})$ depth for one algorithm
(Section~\ref{sec:parallel-rounds}). An implicit interval representation
avoids constructing the dense conflict graph and uses polylogarithmic
work per update. Together with near-linear GD mixing, this gives
near-linear total work.

\paragraph{Counting through learned Jerrum--Sinclair dynamics.}
To improve the $\wtO_\lambda(mn/\varepsilon^2)$ annealing baseline,
we seek observations costing $\wtO_\lambda(n)$ rather than
$\wtO_\lambda(m)$. JS exchanges can replace an occupied edge directly,
but uniform proposals still waste steps on dense regions while rarely
selecting isolated edges. We therefore learn nonuniform rates from
monomer marginals. An observable-specific variance bound lets ordinary
GD estimate all these marginals within constant factors in
$\wtO_\lambda(m+n)$ work, even when they are small
(Lemma~\ref{js:lem:learning}). The estimates weight the
classical JS insertion, deletion, and exchange moves, with a
Metropolis--Hastings correction preserving the matching law
\cite{JerrumSinclair,JerrumBook,Hastings70}.

We compare the learned chain with star updates, which resample the
edges incident to one vertex. On bipartite graphs, adding a blank for
each unmatched vertex on one side gives a fixed-cardinality encoding
whose down--up walk resamples stars on that side. Our joint monomer--dimer
covariance bound controls this encoding under positive pinning.
Its proof adapts the covariance, harmonic-surrogate, and random-revealing
arguments of \cite{ConcurrentHolant26}, retaining cancellation through
a corrected monomer statistic (Appendix~\ref{sec:covariance-proof}).
The coefficient $1+O(\sqrt\lambda)$ at small activity gives a down--up
gap at $\lambda=\Theta(\log^{-2}n)$ by \cite{AL20}.
Field-dynamics comparison \cite[Theorem 1.16]{CCCYZ25} transfers
star relaxation to the target activity; an analytic random-bipartition
comparison, controlled by the GD gap, extends it to general graphs.

Star relaxation and tail bounds for the learned rates yield a weak
Poincar\'e inequality and $\wtO_\lambda(n)$ work to mix from
polynomially warm starts (Lemma~\ref{js:lem:learned-mixing}).
This is not a worst-start guarantee for the learned chain.
Thermodynamic integration \cite{GelmanMeng98} expresses $\log Z_G$
through mean matching cardinalities. The bound $\Var|M|=O(n)$
controls discretization and single-observation variance. We control correlations along
one annealing trajectory by applying warm-start mixing to stationary
laws and laws biased by $|M|+1$; the intermediate distributions
themselves need not be warm. Reusing each guide table over a doubling
interval avoids repeated graph scans. With
$\wtO_\lambda(n/\varepsilon^2)$ observation blocks, the total work is
$\wtO_\lambda(n^2/\varepsilon^2)$, including learning.

\subsection{Related work}\label{sec:related-work}

\paragraph{Sampling matchings.}
At fixed activity and accuracy, classical analyses give
$O_\lambda(mn^2\log n)$ mixing for lazy JS dynamics
\cite{JerrumSinclair,Sinclair92}, while GD mixes in
$O_{\lambda,\Delta}(m\log n)$ steps on bounded-degree graphs
\cite{CLV21}. More recent bounds are $\wtO_\lambda(\Delta^2m)$ for JS
and $\wtO_\lambda(\Delta^3m)$ for GD \cite{CFJMYZ25}; both retain
polynomial dependence on the maximum degree.
The two chains differ in their local moves: JS can exchange an occupied
edge with an incident edge, whereas GD only resamples one edge indicator.
A different approach first samples the monomer set and then a perfect
matching on the remaining vertices. Efficient perfect-matching counting
makes this approach implementable on planar graphs \cite{AASV21}.
The method also handles matchings of a prescribed size; it does not
simulate the single-edge GD considered here.

Concurrent work of Xiaoyu Chen, Zejia Chen, and Xinyuan Zhang
\cite{ConcurrentHolant26} establishes degree-free spectral independence
for log-concave Holant measures and an $O_\lambda(m)$ Glauber relaxation
bound for matchings. Linear relaxation does not by itself give
near-linear worst-start mixing, since the standard spectral-gap estimate
also depends on the smallest stationary probability.
We use their covariance estimates, together with
the low-activity estimate of \cite{CCCYZ25}, to verify our general
criterion (Proposition~\ref{prop:holant-edge}). It strengthens this
spectral control to an inverse LSI constant $O_\lambda(m\log n)$ and
near-linear mixing. The same mixing bound holds for JS by comparison
(Corollary~\ref{cor:original-js}). On a disjoint union of edges, both
uniformly edge-selected chains require $\Omega(m\log m)$ updates,
so our worst-case mixing bound is optimal up to a logarithmic factor.
Our counting proof also adapts their
covariance arguments to control monomer and dimer indicators jointly.

\paragraph{General mixing criteria.}
Spectral independence \cite{ALO20}, field dynamics \cite{CFYZ21},
and localization \cite{CE22,CCYZ25,CJMYZ26} connect local spectral
control to global mixing. Criteria that also control entropy include
spectral independence with marginal bounds \cite{CLV21,BCCPSV22},
entropic independence \cite{AJKPV22,JPV26}, and spectral independence
with relative marginal stability \cite{CFYZ22}.
These criteria differ in their pinning assumptions and the dynamics
they analyze: for example, the down--up walks in \cite{AJKPV22} need
not be single-site GD. Section~\ref{sec:mixing-criterion} gives a
more detailed comparison.
These are alternative sufficient conditions, rather than a single
hierarchy in which each criterion implies the others.

Our criterion combines field-dynamics spectral stability along scalar
decreasing fields with marginal assumptions to obtain LSI for ordinary
GD. Its logarithmic dependence on inverse occupied marginals allows
those marginals to be polynomially small. Thus the matching application
does not require degree-independent lower bounds on edge marginals.
Low-degree decompositions also appear in analyses of higher-order
walks \cite{DDFH24,KO20}; our argument uses low-degree approximation to
control single-site GD on variable-cardinality supports.

\paragraph{Parallel sampling.}
The matching sampler in \cite{HKLYZ26}, based on the parallel simulation
of Liu and Yin \cite{LY22,LY25}, has $\wtO_\lambda(\Delta^4)$ depth and
$\wtO_\lambda(m\Delta)$ processors. Substituting our mixing bound gives
$\wtO_\lambda(\Delta)$ depth, but the direct line-graph implementation
does not establish near-linear work. Our interval evaluator achieves
near-linear total work together with this depth bound; the same
implementation also gives both degree-independent depth bounds.
Other approaches use block updates \cite{Lee24}, localization
\cite{CLYZ25}, or entropic independence and counting \cite{ABTV23}.
The localization-based GD simulation in \cite{CLYZ25} still requires
coupling contraction on a suitable high-probability set; rapid mixing
alone does not verify this condition.
The work \cite{ABTV23} also gives a separate planar-perfect-matching
sampler, whose depth was subsequently improved in \cite{AGR24}.
The matching sampler of Feng and Yin \cite{FY18} uses the LOCAL model,
where local computation and message length are not charged.

General oracle-based sampling methods \cite{AGR24,ABCHKLV26} bound
the expected number of adaptive rounds and conditional-marginal queries.
A query is treated as a primitive, so these bounds do not by themselves
account for the cost of implementing it.
Continuous-walk reductions \cite{AHLVXY23,ACV24} require both
weighted-counting access and covariance or transport control under
arbitrary exponential tilts. Our scalar decreasing-field estimates
do not supply those hypotheses. In contrast, our parallel sampler
accounts for all work and gives a high-probability depth bound without
a counting oracle.

\paragraph{Counting and learned dynamics.}
Warm starts and dependent observations are already part of classical
annealing \cite[Section 7]{SVV09}.
Recent subquadratic counting algorithms for vertex-spin models use
low-variance marginal estimation \cite{AFFGW25} or aggregation of perfect
marginal samples \cite{CCLZ26}, on bounded-degree or restricted-growth
graphs. Our algorithm counts matchings, an edge model, without a degree
restriction or access to a perfect-marginal oracle.
On dense graphs, it saves a factor $n$ over the independent-sample
annealing baseline, up to logarithms. The analysis combines thermodynamic
integration \cite{GelmanMeng98} with warm-start mixing of learned JS
dynamics, rather than the paired-product estimator of \cite{HKLYZ26}.

Learning sampling parameters also has precedents. The algorithm of
\cite{JSV04} learns weights balancing perfect and near-perfect bipartite
matching classes during annealing. Marginal-based domain sparsification
\cite{AD20,ADVY22} instead uses preprocessing to accelerate subsequent
samples from the same law. The learned quantities play a different role
here: monomer-marginal estimates determine local JS move rates, while
the Metropolis--Hastings correction \cite{Hastings70} preserves the
original matching law. Thus our dynamics do not reweight auxiliary
matching classes. Preservation of the target law is also a feature
of the sparsification methods.

\section{Preliminaries}\label{sec:preliminaries}\label{sec:algorithms}

We retain the notation introduced in Section~\ref{sec:intro}.
All logarithms are natural unless a base is indicated.

For matching laws, a feasible edge pinning conditions specified edges
to be present or absent.
Deleting the endpoints of edges pinned present and the edges pinned
absent leaves a simple graph on which the residual law is again a
matching law at activity $\lambda$.

\subsection{Functional inequalities and mixing}\label{sec:functional-inequalities}

We use the entropy and Dirichlet form introduced in
Section~\ref{sec:mixing-criterion}. Let $P$ be an irreducible
$\mu$-reversible Markov kernel on a finite support $\Omega$ with at least
two states. In operator notation,
\[
 \cE_P(f,g)=\langle f,(\mathrm{Id}-P)g\rangle_\mu.
\]
Write $\Var_\mu(f)=\E_\mu[f^2]-(\E_\mu f)^2$ and
$\operatorname{osc}(f)=\max_\Omega f-\min_\Omega f$.
The spectral gap $\gamma$, log--Sobolev constant $\alpha_{\rm LS}$,
and modified log--Sobolev constant $\alpha_{\rm mLS}$ are the largest
constants for which the following inequalities hold, respectively:
\begin{align}
 \Var_\mu(f)&\le\gamma^{-1}\cE_P(f,f),
 &&f:\Omega\to\mathbb R,\notag\\
 \Ent_\mu(f^2)&\le\alpha_{\rm LS}^{-1}\cE_P(f,f),
 &&f:\Omega\to\mathbb R,\label{eq:lsi}\\
 \Ent_\mu(h)&\le\alpha_{\rm mLS}^{-1}\cE_P(h,\log h),
 &&h>0.\label{eq:mlsi}
\end{align}
These are the Poincar\'e, log--Sobolev (LSI), and modified
log--Sobolev (MLSI) inequalities.
The pointwise inequality
$(a-b)(\log a-\log b)\ge4(\sqrt a-\sqrt b)^2$ for $a,b>0$
gives $\alpha_{\rm mLS}\ge4\alpha_{\rm LS}$ with these conventions;
no additional assumption is needed. See \cite{DSC96} for background
on log--Sobolev inequalities for finite Markov chains.

Total variation and worst-start mixing time are as defined in the
introduction. For kernels with nonnegative spectrum and
$0<\varepsilon\le1/2$, the log--Sobolev inequality gives
\begin{equation}\label{green-eq:ls-mixing}
 t_{\rm mix}(P,\varepsilon)
 \lesssim
 \alpha_{\rm LS}^{-1}\log\log\frac{e^2}{\mu_{\min}}
 +\gamma^{-1}\log\frac1\varepsilon.
\end{equation}
Hypercontractivity of $e^{t(P-\mathrm{Id})}$ first brings a point-mass
density to bounded $L^2$ norm. Nonnegative spectrum gives
$\|P^k h\|_2\le\|e^{k(P-\mathrm{Id})}h\|_2$; spectral-gap contraction
then controls the remaining error.
See \cite[Theorem 3.7 and Corollary 3.8]{DSC96}.
Glauber dynamics has nonnegative spectrum because its kernel is an
average of conditional-expectation projections. We count individual
coordinate updates, not sweeps.

\subsection{Field dynamics and spectral-gap comparison}
\label{sec:field-preliminaries}

We record the analytic field-dynamics construction and the known
Poincar\'e bound that motivates Theorem~\ref{thm:intro-boolean}.
For a downward-closed law $\mu$ on $2^{[N]}$ and a field
$z\in(0,\infty)^N$, define
\[
 (z*\mu)(S)\propto\mu(S)\prod_{i\in S}z_i.
\]
A scalar field has every coordinate equal. Positive pinnings and the
laws $\nu_t^\tau=((1-t)*\mu)^\tau$ are as in the introduction.

For $0<a<1$, a transition of field dynamics $F_a$ retains each occupied
coordinate with probability $1-a$, producing $\tau$, and draws the
remaining coordinates from $(a*\mu)^\tau$, restoring $\tau$ to obtain
the full configuration. By the posterior identity in
Section~\ref{sec:mixing-criterion}, conditional on the revealed set
$Y_{1-a}$ its input and output are independent draws from the same
posterior. It is therefore reversible, with $X\sim\mu$ and
\begin{equation*}
 \cE_{F_a}(f,f)=
 \E\bigl[\Var(f(X)\mid Y_{1-a})\bigr].
\end{equation*}
This construction is used for analysis, not as an implemented sampling
oracle in our algorithms.

\begin{proposition}[Known spectral-gap bound
{\cite{CE22,CCYZ25,CCCYZ25}}]\label{green-lem:known-gap}
Let $\mu$ have nontrivial downward-closed support in $2^{[N]}$ and satisfy
the field-dynamics spectral-stability and bounded-marginal-ratio assumptions
(i)--(ii) of Theorem~\ref{thm:intro-boolean}. Then
\begin{equation*}
 \gap(P_{\mathrm{GD}})^{-1}\le N(1+\kappa)e^I.
\end{equation*}
No occupied-marginal lower bound is needed.
\end{proposition}
\begin{proof}
Apply \cite[Theorem 1.16 and Lemma 1.18]{CCCYZ25}: the field gap
satisfies $\liminf_{a\downarrow0}\gap(F_a)/a\ge e^{-I}$,
and comparison with one Glauber update costs at most $N(1+\kappa)$.
\end{proof}
Spectral-gap contraction alone gives a mixing bound with a logarithmic
startup dependence on $1/\mu_{\min}$.
Theorem~\ref{thm:intro-boolean} strengthens this to log--Sobolev control,
reducing that dependence to a double logarithm via
\eqref{green-eq:ls-mixing}.

\subsection{Computational model}\label{sec:computational-model}

The input is a labeled edge list, with $O(m+n)$
initialization work. We count arithmetic operations, comparisons,
exact uniform choices, and Bernoulli draws in the random-choice RAM
model.

For parallel computation we use a CREW PRAM, with work equal to the total
number of operations and depth equal to the longest dependency chain.
Sorting, scans, reductions, and compaction have polylogarithmic depth
and linear work up to logarithmic factors; sorting networks suffice
\cite{Batcher68}. These are arithmetic-work guarantees, not finite-word or
random-bit complexity bounds.

\section{Glauber mixing from field-dynamics spectral stability}
\label{sec:abstract-glauber}\label{sec:hierarchy}

We prove \cref{thm:intro-boolean} using a low-degree approximation
criterion for LSI. Section~\ref{sec:abstract-entropy} establishes this
criterion and reduces the theorem to two supporting estimates, proved
in Sections~\ref{sec:abstract-hierarchy} and~\ref{abs-sec:difference}.

\subsection{A low-degree approximation criterion for LSI}
\label{sec:abstract-entropy}

Use the notation of \cref{thm:intro-boolean}, with $I$ as in
\eqref{eq:integrated-stability}.
Write $r=\max_{S\in\Omega}|S|\ge1$ and $S-i=S\setminus\{i\}$. We
use the residual configuration spaces $\Omega^\tau$ defined in the introduction.

We approximate functions on $\Omega$ by multilinear polynomials in
the occupation indicators $X_i(S)=\one_{\{i\in S\}}$.
For $A\in\Omega$, let
$X_A(S)=\prod_{i\in A}X_i(S)=\one_{\{A\subseteq S\}}$, and define
\[
 \cV_j:=\operatorname{span}\{X_A:A\in\Omega,\ |A|\le j\},
 \qquad 0\le j\le r.
\]
Thus $\cV_j$ consists of functions admitting a polynomial representation
of degree at most $j$ on $\Omega$; $j$ is the approximation cutoff.
Let $\Pi_j:L^2(\mu)\to\cV_j$ denote the orthogonal projection, defined by
\[
 \Pi_j f:=\underset{g\in\cV_j}{\operatorname{argmin}}\;
             \E_\mu[(f-g)^2].
\]
This is the unique best mean-square approximation to $f$ in $\cV_j$,
characterized by the orthogonal decomposition
\[
 f=\Pi_jf+(f-\Pi_jf),\qquad
 \langle f-\Pi_jf,h\rangle_\mu=0\quad(h\in\cV_j).
\]
Since monomials need not be orthogonal under $\mu$, this is not simply
a truncation of their coefficients.

Order $\Omega$ by nondecreasing cardinality. The evaluation matrix
$(X_A(S))_{S,A\in\Omega}$ is lower triangular with ones on its diagonal:
$A\subseteq S$ implies $|A|\le|S|$, with equality only when $A=S$.
Hence these monomials form a basis, and $\cV_r=L^2(\mu)$.

Unmarked norms and inner products use $\mu$; in particular,
$\|f\|_2^2=\E_\mu f^2$ and $\|f\|_\infty=\max_{S\in\Omega}|f(S)|$.
With one coordinate update as the time unit,
\[
\cE_{\mathrm{GD}}(f,f)=N^{-1}\sum_i\E_\mu[\Var_\mu(f\mid X_{-i})].
\]

For the uniform law on $2^{[N]}$, Fourier analysis gives a prototype
of the Dirichlet-form bound below. The orthonormal characters
$\chi_A=\prod_{i\in A}(2X_i-1)$ have degree $|A|$. Writing
$f=\sum_A\widehat f(A)\chi_A$, the standard identity
\cite[Theorem 2.38]{ODonnell14} gives
\[
 \cE_{\mathrm{GD}}(f,f)
 =\frac1N\sum_{A\subseteq[N]}|A|\widehat f(A)^2
 \ge\frac{j+1}{N}\|f-\Pi_jf\|_2^2,
 \qquad 0\le j<N.
\]
Thus the part beyond degree $j$ relaxes at a rate at least $(j+1)/N$.
For dependent distributions, GD need not preserve the low-degree polynomial
spaces. Our proof therefore controls the error of low-degree approximation directly.

The following criterion uses the classical super-Poincar\'e approach to
LSI \cite{Wang00SPI,GM15}. We give a direct proof for arbitrary nested
approximation spaces.

\begin{proposition}[Low-degree approximation criterion for LSI]
\label{abs-prop:degree-to-functional}
Let $\mu$ have finite support $\Omega$ with at least two states, and let
$P$ be an irreducible $\mu$-reversible Markov kernel.
Let $r\ge1$ be integer, and let
\[
 \{\text{constant functions}\}=\cV_0\subseteq\cV_1
 \subseteq\cdots\subseteq\cV_r=L^2(\mu)
\]
be nested linear subspaces with orthogonal projections $\Pi_j$.
Suppose that, for some $A,\Lambda\ge1$, the following bounds hold:
\begin{enumerate}[label=\textup{(\roman*)},leftmargin=*]
\item \textup{\textbf{Dirichlet-form control of approximation error.}}
For $0\le j<r$ and $f:\Omega\to\mathbb R$,
\[
 \cE_P(f,f)\ge\frac{j+1}{A}\|(\mathrm{Id}-\Pi_j)f\|_2^2.
\]
\item \textup{\textbf{Norm bound on the approximation spaces.}}
For $0\le j\le r$ and $g\in\cV_j$,
\[
 \|g\|_\infty\le e^{\Lambda j/2}\|g\|_2.
\]
\end{enumerate}
Then $\alpha_{\rm LS}^{-1}\lesssim A\Lambda$.
All implicit constants are universal.
\end{proposition}

To see how these assumptions imply LSI, consider a function supported on a rare event.
At a cutoff chosen according to the event's probability, the norm bound
limits the $L^2$ mass captured by its projection. A substantial residual
remains, so the Dirichlet-form bound forces a large value relative to
the function's squared $L^2$ norm. Rarer events permit larger cutoffs,
yielding a logarithmic penalty in the inverse event probability.
Applying this estimate to truncated levels of a general function,
and then centering, gives LSI.

\begin{proof}[Proof of \cref{abs-prop:degree-to-functional}]
\emph{Dirichlet-form bound for functions supported on rare events.}
Let $g\ne0$ be supported on a set of measure $s\le1/4$.
Testing $\Pi_jg$ against unit vectors in $\cV_j$, the norm bound
and Cauchy--Schwarz give
\[
 \|\Pi_jg\|_2\le e^{\Lambda j/2}\|g\|_1
                 \le e^{\Lambda j/2}\sqrt{s}\,\|g\|_2.
\]
Choose $j=\lfloor\log(1/s)/(2\Lambda)\rfloor$. This cutoff is
admissible: since $\cV_r=L^2(\mu)$, applying the norm bound to
each state indicator gives $1\le e^{\Lambda r/2}\sqrt{\mu(x)}$.
Thus $\mu_{\min}\ge e^{-\Lambda r}$, so $s\ge\mu_{\min}$ implies
$0\le j<r$. Moreover, $s e^{\Lambda j}\le\sqrt{s}\le1/2$:
at least half of $g$'s squared $L^2$ norm lies outside $\cV_j$.
The Dirichlet-form bound therefore yields
\begin{equation}\label{abs-eq:small-support}
 \cE_P(g,g)\ge\frac{j+1}{A}
                (1-s e^{\Lambda j})\|g\|_2^2
 \ge\frac{\log(1/s)}{4A\Lambda}\|g\|_2^2.
\end{equation}
This is the logarithmic rare-support penalty used below.

\smallskip\noindent\emph{From the rare-support bound to entropy.}
We decompose an arbitrary function into levels supported on rare events,
with truncations whose Dirichlet-form values can be summed without loss.
Normalize $\|f\|_2=1$ and truncate $F=|f|$:
$g_k=\min\{(F-2^k)_+,2^k\}$, $k\ge1$. Each nonzero $g_k$ has support of measure
at most $4^{-k}$, so \eqref{abs-eq:small-support} gives
$k\|g_k\|_2^2\lesssim A\Lambda\cE_P(g_k,g_k)$.
For $2^{k+1}\le F<2^{k+2}$, $g_k=2^k$ and
$F^2\log F^2\lesssim k g_k^2$; the part $F<4$ contributes $O(1)$.
Summing over the dyadic levels gives
\[
 \Ent_\mu(f^2)\lesssim1+\sum_{k\ge1}k\|g_k\|_2^2
 \lesssim1+A\Lambda\sum_{k\ge1}\cE_P(g_k,g_k).
\]
It remains to control the sum of the truncated Dirichlet forms.
For each pair $x,y$, the differences
$|g_k(x)-g_k(y)|$ are lengths of disjoint slices of the interval
between $|f(x)|$ and $|f(y)|$. Hence
$\sum_k|g_k(x)-g_k(y)|^2\le|f(x)-f(y)|^2$ for every pair $x,y$.
Summing against the transition weights gives
$\sum_k\cE_P(g_k,g_k)\le\cE_P(f,f)$. Combining these estimates and rescaling proves
\[
 \Ent_\mu(f^2)\lesssim A\Lambda\cE_P(f,f)+\|f\|_2^2.
\]
To remove the extra $L^2$ term, apply this estimate to $h=f-\mu f$.
Rothaus' inequality~\cite{Rothaus} then gives
\[
 \Ent_\mu(f^2)\le\Ent_\mu(h^2)+2\Var_\mu(f)
 \lesssim A\Lambda\cE_P(f,f)+\Var_\mu(f)
 \lesssim A\Lambda\cE_P(f,f).
\]
The last step uses the $j=0$ Dirichlet-form bound
$\Var_\mu(f)\le A\cE_P(f,f)$ and $\Lambda\ge1$.
Thus $\alpha_{\rm LS}^{-1}\lesssim A\Lambda$.
\end{proof}

\paragraph{The Dirichlet-form and low-degree norm bounds.}
For the polynomial spaces defined above, the following lemmas verify
the two hypotheses of \cref{abs-prop:degree-to-functional}.
Field-dynamics spectral stability supplies the Dirichlet-form bound;
the marginal assumptions supply the low-degree norm bound. Their proofs are
given in Sections~\ref{sec:abstract-hierarchy}
and~\ref{abs-sec:difference}, respectively.

\begin{lemma}[Dirichlet-form bound]\label{abs-thm:hierarchy}
Assume the field-dynamics spectral-stability and integrability hypotheses
\eqref{intro-eq:stability} and \eqref{eq:integrated-stability},
and the marginal-ratio bound \eqref{intro-eq:kappa} with $\kappa>0$.
Then, for every $0\le j<r$ and every $f:\Omega\to\mathbb R$,
\begin{equation}\label{abs-eq:hierarchy}
 \cE_{\mathrm{GD}}(f,f)
 \ge\frac{(j+1)e^{-I}}{N(1+\kappa)}\|f-\Pi_jf\|_2^2.
\end{equation}
\end{lemma}

\begin{lemma}[Low-degree norm bound]
\label{abs-lem:christoffel-rho-b}
Assume \eqref{intro-eq:kappa} and \eqref{intro-eq:b}. For every positive
pinning $\nu=\mu^\tau$, every integer $j\ge0$, and every polynomial
$f$ of degree at most $j$ on $\Omega^\tau$,
\[
 \|f\|_{\infty,\nu}
 \le\left(1+N\sqrt{\frac{1+\kappa}{b}}\right)^j\|f\|_{2,\nu}.
\]
\end{lemma}

\paragraph{Completing the mixing criterion.}
It remains to insert the two estimates and convert LSI to mixing,
keeping the spectral-gap bound for the accuracy term.

\begin{proof}[Proof of \cref{thm:intro-boolean}]
Apply \cref{abs-prop:degree-to-functional} using
\cref{abs-thm:hierarchy,abs-lem:christoffel-rho-b}, with
\[
 A=N(1+\kappa)e^I,\qquad
 \Lambda=2\log\left(1+N\sqrt{\frac{1+\kappa}{b}}\right).
\]
Since $\kappa\ge1$, $A\lesssim N\kappa e^I$ and
$\Lambda\lesssim\log(N\kappa/b)$. The latter comparison uses
$b\le\kappa/(1+\kappa)$: for any feasible coordinate $i$, summing
$\mu(T+i)\le\kappa\mu(T)$ over legal insertions gives
$\mu(X_i=1)\le\kappa\mu(X_i=0)$, while
$\mu(X_i=1)\ge b$. In particular, $N\kappa/b\ge2$.
This proves the stated LSI. Moreover, the Dirichlet-form bound at $j=0$ gives
$\gamma^{-1}\le A$ for $\gamma=\gap(P_{\mathrm{GD}})$.
Glauber dynamics has nonnegative spectrum, so substituting these
LSI and gap bounds into \eqref{green-eq:ls-mixing} proves
\eqref{abs-eq:mix} for every $0<\varepsilon\le1/2$.

For the final asymptotic assertion, the state-indicator argument in the
proof of \cref{abs-prop:degree-to-functional} gives
$\log(1/\mu_{\min})\le r\Lambda\le N\Lambda$.
Under the stated bounds on $\kappa,I,b$, we have
$A=O(N)$ and $\Lambda=O(\log N)$, hence inverse LSI constant
$O(N\log N)$. Also
$\log\log(e^2/\mu_{\min})\le\log(2+N\Lambda)=O(\log N)$,
giving mixing time $O(N[\log^2N+\log(1/\varepsilon)])$.
\end{proof}

For fixed $\kappa$ and $I$, the lower bound $b$ enters only
logarithmically. This dependence is necessary for LSI even for a
Bernoulli law with $\mu(1)=b\in(0,1)$: testing
$f=\one_{\{1\}}$ gives $\alpha_{\rm LS}^{-1}\ge\log(1/b)/(1-b)$,
although its field-dynamics spectral-stability rate can be taken to be
$C(t)\equiv1$.

\subsection{Dirichlet-form bounds from field-dynamics spectral stability}
\label{sec:abstract-hierarchy}

We prove \cref{abs-thm:hierarchy}, the Dirichlet-form bound used in
Section~\ref{sec:abstract-entropy}. The factor $j+1$ strengthens the
Poincar\'e inequality, recovered at $j=0$, to control the error of
approximation at every degree.

Consider continuous-time \emph{insertion--deletion dynamics} on $\Omega$:
from $S$, delete each $i\in S$ at rate one and insert each
$i\notin S$ with $S+i\in\Omega$ at rate $\mu(S+i)/\mu(S)$.
The insertion and deletion rates satisfy detailed balance, so
the dynamics is reversible with respect to $\mu$. Its Dirichlet form is
\[
 \cE_1(f,f)
 =\sum_{S\in\Omega}\mu(S)\sum_{i\in S}[f(S)-f(S-i)]^2.
\]
Each legal pair $(S,S+i)$ has weight
$\mu(S+i)/[N(1+\mu(S+i)/\mu(S))]$ in $\cE_{\mathrm{GD}}$.
The marginal-ratio bound \eqref{intro-eq:kappa} therefore gives
\[
 \cE_1(f,f)\le N(1+\kappa)\cE_{\mathrm{GD}}(f,f).
\]
It suffices to prove the following auxiliary estimate:
\begin{equation}\label{abs-eq:bdhier}
 \cE_1(f,f)
 \ge(j+1)e^{-I}\|f-\Pi_jf\|_2^2,
 \qquad 0\le j<r,
\end{equation}
which implies \eqref{abs-eq:hierarchy} by the comparison above.
We prove \eqref{abs-eq:bdhier} using only spectral stability and
integrability. The unit deletion rates make this form convenient along
the down walk: the deletion operator stays fixed while the stationary law changes.

\paragraph{Interpolating Dirichlet forms along the down walk.}
We use the continuous-time down walk underlying the localization
representation of field dynamics~\cite[Section~3.1.2]{CCYZ25}.
In this representation, a field-dynamics transition retains a random
subset of occupied coordinates and then resamples all remaining
coordinates from the posterior given the retained set
(Section~\ref{sec:field-preliminaries}).
These posteriors are the decreasing-field laws controlled by spectral
stability. We use this covariance control to compare the Dirichlet forms
of auxiliary insertion--deletion dynamics whose stationary laws are the
retained-set distributions. As retention vanishes, normalized monomials
become orthogonal and their Dirichlet-form values approach their degrees,
providing an initial bound that we can propagate back to the original law.

The \emph{continuous-time down walk} deletes the elements of $X\sim\mu$
at independent uniform times in $[0,1]$.
Let $X_\theta$ be the remaining set at time $1-\theta$, so each element
is retained independently with probability $\theta$; the down direction
is $\theta\downarrow0$. For $0<\theta\le1$, write
\begin{equation*}
 \rho_\theta(S)=\Prb(X_\theta=S)
 =\theta^{|S|}\sum_{B\in\Omega^S}(1-\theta)^{|B|}\mu(S\cup B).
\end{equation*}
Differentiating gives
\begin{equation}\label{abs-eq:thinning-derivative}
 \theta\rho_\theta'(S)=|S|\rho_\theta(S)
                         -\sum_{k\notin S}\rho_\theta(S+k).
\end{equation}
Indeed, differentiating $(1-\theta)^{|B|}$ counts each nonempty $B$
$|B|$ times, once for each $k\in B$.

Consider auxiliary insertion--deletion dynamics that deletes each occupied coordinate
at rate one and inserts each addable coordinate $i$ at rate
$\rho_\theta(S+i)/\rho_\theta(S)$. Its stationary law is $\rho_\theta$,
and its Dirichlet form (not that of the down walk), for
$g,h:\Omega\to\mathbb R$, is
\begin{equation}\label{abs-eq:thinned-energy}
 \cE_\theta(g,h)=\sum_S\rho_\theta(S)\sum_{i\in S}
 [g(S)-g(S-i)][h(S)-h(S-i)].
\end{equation}
Reversibility combines the insertion and deletion contributions into this single sum.
Its generator $L_\theta$, with deletion part $\mathcal D$, is
\[
 \mathcal Dh(S)=\sum_{i\in S}[h(S-i)-h(S)],\qquad
 L_\theta h(S)=\mathcal Dh(S)
 +\sum_{i\notin S}\frac{\rho_\theta(S+i)}{\rho_\theta(S)}
                         [h(S+i)-h(S)].
\]
Insertion sums use only legal insertions, and probabilities of
infeasible configurations are zero. Then
$\cE_\theta(g,h)=-\langle g,L_\theta h\rangle_{\rho_\theta}$.

\paragraph{A variational formulation of the Dirichlet-form bound.}
Fix $0\le j<r$. Let $\Pi_{j,\theta}$ be the orthogonal projection onto
$\cV_j$ in $L^2(\rho_\theta)$, and write
\begin{equation}\label{abs-eq:least-squares}
 Q_\theta(f):=\|f-\Pi_{j,\theta}f\|_{2,\rho_\theta}^2
       =\min_{g\in\cV_j}\Var_{\rho_\theta}(f-g).
\end{equation}
The second equality holds because $\cV_j$ contains constants. Define
\begin{equation}\label{abs-eq:best-constant}
 \gamma_j(\theta):=\inf_{f\notin\cV_j}
                       \frac{\cE_\theta(f,f)}{Q_\theta(f)}.
\end{equation}
This is the largest coefficient in the inequality
$\cE_\theta(f,f)\ge\gamma_j(\theta)Q_\theta(f)$; at $j=0$ it is
the spectral gap of the auxiliary insertion--deletion dynamics.
Our goal is exactly $\gamma_j(1)\ge(j+1)e^{-I}$.
We will prove that $\liminf_{\theta\downarrow0}\gamma_j(\theta)\ge j+1$,
and that, as $\theta$ increases toward $1$, its relative rate of decrease
is at most $(C(\theta)-1)/(1-\theta)$.
Integrating then gives the desired bound.
The next lemma supplies the derivative estimate. Its paired-increment
calculation is related to the discrete Bochner estimates in
\cite[Lemmas 2.2--2.3]{Bochner} and earlier jump-process Bochner
identities \cite{BCDP06,KKO13}. Those works establish spectral-gap
estimates; here we use the calculation along the down walk to control
approximation error at every polynomial cutoff.

\begin{lemma}[Dirichlet-form evolution along the down walk]\label{abs-lem:ss-sum}
Assume field-dynamics spectral stability \eqref{intro-eq:stability}
with rate $C$. For $0<\theta<1$, set
\begin{equation}\label{abs-eq:beta}
 \beta(\theta):=\frac{\theta(C(\theta)-1)}{1-\theta}.
\end{equation}
Then every fixed function $h:\Omega\to\mathbb R$ satisfies
\begin{equation}\label{abs-eq:energy-bound}
 \theta\frac d{d\theta}\cE_\theta(h,h)
 \ge\|\mathcal Dh\|_{2,\rho_\theta}^2
       -\|(L_\theta-\mathcal D)h\|_{2,\rho_\theta}^2
       -\beta(\theta)\cE_\theta(h,h).
\end{equation}
\end{lemma}

Assuming this estimate, we first complete the proof of
\cref{abs-thm:hierarchy}, then verify the estimate itself.

\begin{proof}[Proof of \cref{abs-thm:hierarchy}]
We prove the bound on $\gamma_j(1)$ outlined above.
Inner products and norms in the calculations below use $\rho_\theta$.

\paragraph{Differentiating at an optimizing function.}
Fix a parameter $\theta$ where $\gamma_j$ is differentiable, and
choose a function $f$ attaining \eqref{abs-eq:best-constant}.
Existence and the differentiation of this minimum are justified below.
Write $g=\Pi_{j,\theta}f$, $R=f-g$, $Q=Q_\theta(f)$, and
$\gamma=\gamma_j(\theta)$. Varying $f$ in a direction $h$ at a minimum
of the quotient gives
\[
 \cE_\theta(f,h)=\gamma\langle R,h\rangle
 \quad(h:\Omega\to\mathbb R).
\]
Indeed, the numerator and denominator have directional derivatives
$2\cE_\theta(f,h)$ and $2\langle R,h\rangle$, respectively.
Consequently,
\begin{equation}\label{abs-eq:optimizer}
 -L_\theta f=\gamma R,\qquad
 \cE_\theta(f,g)=0,\qquad
 \cE_\theta(f,f)=\gamma Q.
\end{equation}
The middle identity follows from $\langle R,g\rangle=0$.
Thus optimality supplies the additional Dirichlet-form orthogonality
needed below; it does not assert that $L_\theta$ preserves $\cV_j$.

Keep this optimizing $f$ fixed when differentiating in $\theta$.
We use $\partial_\theta$ to emphasize this convention: the law and
projection vary, but the function does not.
For a fixed $h$, \eqref{abs-eq:thinning-derivative} gives
\[
 \theta\partial_\theta\|h\|_2^2
 =-\cE_\theta(h,h)-2\langle h,\mathcal Dh\rangle.
\]
This follows by reindexing the insertion sum and using
$x^2-y^2=-(x-y)^2+2x(x-y)$.
When differentiating $Q_\theta(f)=\langle R,R\rangle_{\rho_\theta}$,
the change in $R=f-\Pi_{j,\theta}f$ contributes zero because
$\partial_\theta\Pi_{j,\theta}f\in\cV_j$ and
$\langle R,\partial_\theta\Pi_{j,\theta}f\rangle_{\rho_\theta}=0$.
The remaining change of measure is given by the preceding identity with $h=R$.
Also $\mathcal D\cV_j\subseteq\cV_j$, since
$\mathcal DX_A=-|A|X_A$, so $\langle R,\mathcal Dg\rangle=0$.
Using \eqref{abs-eq:optimizer}, we obtain
\begin{equation}\label{abs-eq:projection-derivative}
 \theta\partial_\theta Q_\theta(f)
 =-\gamma Q-\cE_\theta(g,g)-2\langle R,\mathcal Df\rangle.
\end{equation}
Here $\cE_\theta(R,R)=\cE_\theta(f,f)+\cE_\theta(g,g)$ because
$\cE_\theta(f,g)=0$.

On the other hand, substitute
$(L_\theta-\mathcal D)f=-\gamma R-\mathcal Df$ into
\cref{abs-lem:ss-sum}. Expanding the squared norm gives
\[
 \theta\partial_\theta\cE_\theta(f,f)
 \ge-\gamma^2Q-2\gamma\langle R,\mathcal Df\rangle
                         -\beta(\theta)\gamma Q.
\]
Subtract $\gamma$ times \eqref{abs-eq:projection-derivative}.
Both the squared-residual term and the mixed term cancel:
\[
 \theta\bigl[\partial_\theta\cE_\theta(f,f)
                    -\gamma\partial_\theta Q_\theta(f)\bigr]
 \ge\gamma\cE_\theta(g,g)-\beta(\theta)\gamma Q
 \ge-\beta(\theta)\gamma Q.
\]
Dividing by $Q$ differentiates the quotient at its minimizer.
Thus, with $\beta$ from \eqref{abs-eq:beta},
\begin{equation}\label{abs-eq:hierarchy-ode}
 \gamma_j'(\theta)\ge-\frac{\beta(\theta)}{\theta}\gamma_j(\theta)
                  =-\frac{C(\theta)-1}{1-\theta}\gamma_j(\theta)
 \quad\text{for almost every }\theta\in(0,1).
\end{equation}

For completeness, we justify attainment and differentiation without
choosing a differentiable family of optimizers. Adding constants changes
neither $\cE_\theta(f,f)$ nor $Q_\theta(f)$, so restrict to the fixed
space $H=\{f:f(\varnothing)=0\}$. The deletion graph is connected;
hence $\cE_\theta$ is positive definite on $H$. Moreover,
$\rho_\theta(S)\ge\theta^{|S|}\mu(S)>0$, and the Gram matrix defining
$\Pi_{j,\theta}$ is smooth and positive definite for $\theta>0$.
Both quadratic forms therefore have smooth coefficients, and
\[
 \gamma_j(\theta)^{-1}
 =\max_{\substack{f\in H\\\|f\|_{\mathrm{Euclidean}}=1}}
                         \frac{Q_\theta(f)}{\cE_\theta(f,f)}.
\]
Here the norm is the usual Euclidean norm of the vector of function values.
The maximum is attained and positive because $\cV_j\ne L^2(\rho_\theta)$.
On compact intervals bounded away from zero, the denominator is uniformly
positive on this fixed unit sphere. The maximum, and hence $\gamma_j$,
is locally Lipschitz. At a differentiability point $\theta$, freeze any
minimizer $f$. For nearby $t$,
$\gamma_j(t)\le\cE_t(f,f)/Q_t(f)$, with equality at $t=\theta$.
The two functions therefore have the same derivative there, which
justifies the quotient differentiation above even if the minimizer is
not unique.

\paragraph{Establishing the Dirichlet-form bound as $\theta\downarrow0$.}
We now identify the initial coefficient $j+1$.
For nonempty $A\in\Omega$, put
\[
 p_A=\Prb(A\subseteq X)>0,\qquad
 \phi_{A,\theta}(S)=\frac{\one_{\{A\subseteq S\}}}
                            {\sqrt{\theta^{|A|}p_A}}.
\]
The identity $\Prb(A\subseteq X_\theta)=\theta^{|A|}p_A$ gives
\begin{align}
 \Cov_{\rho_\theta}(\phi_{A,\theta},\phi_{B,\theta})
 &=\theta^{|A\triangle B|/2}\frac{p_{A\cup B}}{\sqrt{p_Ap_B}}
   -\theta^{(|A|+|B|)/2}\sqrt{p_Ap_B}
   \longrightarrow\one_{\{A=B\}},\label{abs-eq:initial-covariance}\\
 \cE_\theta(\phi_{A,\theta},\phi_{B,\theta})
 &=|A\cap B|\theta^{|A\triangle B|/2}
                  \frac{p_{A\cup B}}{\sqrt{p_Ap_B}}
   \longrightarrow |A|\one_{\{A=B\}}\label{abs-eq:initial-energy}
\end{align}
as $\theta\downarrow0$, with $p_{A\cup B}=0$ for infeasible unions.
Thus the covariance matrix tends to the identity and the
Dirichlet-form matrix tends to the diagonal matrix of degrees.
This explains why the first degree beyond the cutoff $j$ supplies
the coefficient $j+1$.

We need this comparison uniformly over functions, since
\eqref{abs-eq:best-constant} takes an infimum. These normalized monomials,
together with constants, form a basis. In fixed dimension, the
entrywise limits \eqref{abs-eq:initial-covariance}--\eqref{abs-eq:initial-energy}
imply convergence in operator norm. Hence there is
$0\le\delta_\theta\to0$ such that every
$h=c+\sum_{A\ne\varnothing}a_A\phi_{A,\theta}$ satisfies, for small $\theta$,
\[
 \cE_\theta(h,h)\ge(1-\delta_\theta)\sum_{A\ne\varnothing}|A|a_A^2,
 \qquad
 Q_\theta(h)\le(1+\delta_\theta)\sum_{|A|>j}a_A^2.
\]
For the second inequality, approximate by the terms of degree at most
$j$ and optimize the constant; this uses the variance formulation in
\eqref{abs-eq:least-squares}. It follows that
\[
 \gamma_j(\theta)\ge(j+1)\frac{1-\delta_\theta}{1+\delta_\theta},
 \qquad
 \liminf_{\theta\downarrow0}\gamma_j(\theta)\ge j+1.
\]

\paragraph{Integrating the loss from spectral stability.}
Local Lipschitz continuity makes $\gamma_j$ locally absolutely
continuous. Integrating \eqref{abs-eq:hierarchy-ode} gives, for
$0<s<\theta<1$,
\[
 \gamma_j(\theta)\ge\gamma_j(s)
       \exp\left(-\int_s^\theta\frac{C(t)-1}{1-t}\,dt\right).
\]
Let $s\downarrow0$ and then $\theta\uparrow1$.
The endpoint bound, integrability, and continuity of the quadratic
forms and their variational constant at $1$ yield
$\gamma_j(1)\ge(j+1)e^{-I}$.
This proves \eqref{abs-eq:bdhier}; the comparison with GD at the
start of the subsection proves \eqref{abs-eq:hierarchy}.
\end{proof}

\begin{proof}[Proof of \cref{abs-lem:ss-sum}]
Only the law varies in this derivative. Applying
\eqref{abs-eq:thinning-derivative} to \eqref{abs-eq:thinned-energy} yields
\begin{equation}\label{abs-eq:energy-derivatives}
 \theta\frac d{d\theta}\cE_\theta(h,h)
 =\sum_S\left[|S|\rho_\theta(S)
                    -\sum_{k\notin S}\rho_\theta(S+k)\right]
           \sum_{i\in S}[h(S)-h(S-i)]^2.
\end{equation}

Fix $0<\theta<1$; norms and inner products below use $\rho_\theta$.
The right-hand side of \eqref{abs-eq:energy-derivatives}
equals $\|\mathcal Dh\|_2^2$ minus a weighted sum of cross products of
insertion increments from a common configuration, as verified below. We first bound these
cross products using conditional spectral stability.
Fix $S$ and condition on $X_\theta=S$. The remaining set has law
$\nu_\theta^S$: explicitly, for $T\in\Omega^S$,
\[
 \Prb(X\setminus S=T\mid X_\theta=S)
 =\frac{(1-\theta)^{|T|}\mu(S\cup T)}
        {\sum_{B\in\Omega^S}(1-\theta)^{|B|}\mu(S\cup B)}
 =\nu_\theta^S(T).
\]
In the sum
\[
 \sum_{i\in X\setminus S}[h(S+i)-h(S)]
 =\sum_{i\notin S}[h(S+i)-h(S)]
                      \mathbf1_{\{i\in X\setminus S\}},
\]
the coefficients $h(S+i)-h(S)$ are fixed. All randomness comes from
the indicators under the conditional law of $X$ given $X_\theta=S$.
This statistic links the posterior covariance to the generator:
its conditional mean, multiplied by $\theta/(1-\theta)$, is
$((L_\theta-\mathcal D)h)(S)$, while its second moment contains products
of insertion increments we need to control. The probability-ratio
identities below make this correspondence explicit.
Spectral stability \eqref{intro-eq:stability} gives
\[
 \Var\!\left(\left.
   \sum_{i\notin S}[h(S+i)-h(S)]\mathbf1_{\{i\in X\setminus S\}}
                          \,\right|X_\theta=S\right)
 \le C(\theta)\,
 \E\!\left[\left.
   \sum_{i\notin S}[h(S+i)-h(S)]^2\mathbf1_{\{i\in X\setminus S\}}
                          \,\right|X_\theta=S\right].
\]
Expand the variance. Since indicators square to themselves, the diagonal
second-moment terms equal the expectation on the right without its
factor $C(\theta)$. Moving them across leaves $C(\theta)-1$.
Also, by the definition of $\rho_\theta$,
\begin{align*}
 \frac{\rho_\theta(S+i)}{\rho_\theta(S)}
 &=\frac{\theta}{1-\theta}
       \Prb(i\in X\setminus S\mid X_\theta=S),\\
 \frac{\rho_\theta(S+i+k)}{\rho_\theta(S)}
 &=\left(\frac{\theta}{1-\theta}\right)^2
       \Prb(i,k\in X\setminus S\mid X_\theta=S)
       \qquad(i\ne k).
\end{align*}
Multiplying the resulting inequality by
$\rho_\theta(S)\theta^2/(1-\theta)^2$ therefore yields
\begin{equation}\label{abs-eq:ss-cross-terms}
\begin{split}
 &\sum_{\substack{i,k\notin S\\i\ne k}}
       \rho_\theta(S+i+k)[h(S+i)-h(S)][h(S+k)-h(S)]\\
 &\quad\le
 \frac{\left[\sum_{i\notin S}\rho_\theta(S+i)
                        [h(S+i)-h(S)]\right]^2}{\rho_\theta(S)}
 +\beta(\theta)\sum_{i\notin S}\rho_\theta(S+i)[h(S+i)-h(S)]^2.
\end{split}
\end{equation}

It remains to identify the sum of the left-hand sides. The cancellation
involves the four configurations $S-i-k,S-i,S-k,S$, connected by the
two orders of deleting $i$ and $k$. Grouping terms by the deleted pair
eliminates the value at $S$, leaving products of insertion increments
from the common lower configuration $S-i-k$. These are precisely the
cross products bounded in \eqref{abs-eq:ss-cross-terms}. For each $S$,
\[
 |S|\sum_{i\in S}[h(S)-h(S-i)]^2
       -\left(\sum_{i\in S}[h(S)-h(S-i)]\right)^2
 =\sum_{\{i,k\}\subseteq S}[h(S-i)-h(S-k)]^2.
\]
Also, replacing $S+k$ by $S$ gives
\[
 \sum_S\sum_{k\notin S}\rho_\theta(S+k)
                          \sum_{i\in S}[h(S)-h(S-i)]^2
 =\sum_S\rho_\theta(S)\sum_{k\in S}
                   \sum_{i\in S\setminus\{k\}}
                          [h(S-k)-h(S-i-k)]^2.
\]
Multiply the first identity by $\rho_\theta(S)$ and sum over $S$,
then subtract the second expression. Each unordered pair
$\{i,k\}\subseteq S$ contributes, after division by $\rho_\theta(S)$,
\begin{align*}
 &[h(S-i)-h(S-k)]^2
  -[h(S-i)-h(S-i-k)]^2-[h(S-k)-h(S-i-k)]^2\\
 &\qquad=-2[h(S-i)-h(S-i-k)][h(S-k)-h(S-i-k)].
\end{align*}
Reindex $S-i-k$ as $S$, replace $2\sum_{\{i,k\}}$ by
$\sum_{i\ne k}$, and use \eqref{abs-eq:energy-derivatives} to obtain
\begin{equation}\label{abs-eq:paired-squares}
\begin{split}
 &\|\mathcal Dh\|_{2,\rho_\theta}^2
 -\theta\frac d{d\theta}\cE_\theta(h,h)\\
 &\qquad=\sum_S\sum_{\substack{i,k\notin S\\i\ne k}}
       \rho_\theta(S+i+k)[h(S+i)-h(S)][h(S+k)-h(S)].
\end{split}
\end{equation}
Sum \eqref{abs-eq:ss-cross-terms} over $S$. Its first term on the right
is $\|(L_\theta-\mathcal D)h\|_2^2$, and reindexing its last sum gives
$\beta(\theta)\cE_\theta(h,h)$. Combining this bound with
\eqref{abs-eq:paired-squares} proves \eqref{abs-eq:energy-bound}.
\end{proof}

\subsection{Low-degree norm bound from marginal assumptions}
\label{abs-sec:difference}

We prove \cref{abs-lem:christoffel-rho-b} by induction on polynomial degree.
A forward difference lowers degree by one. The marginal assumptions
control its $L^2$ norm under positive pinning, allowing induction to
bound every local increment. A path of legal insertions and deletions
then bounds the value of the function at any configuration.
No spectral stability is needed.

Both marginal assumptions are inherited by positive pinning, so the
same estimates apply throughout the induction.

\begin{proof}[Proof of \cref{abs-lem:christoffel-rho-b}]
We prove the assertion simultaneously over all positive pinnings
$\nu=\mu^\tau$. For $j=0$, the function is constant and the claim is
immediate.

\paragraph{One difference lowers the degree.}
For an unpinned coordinate $i$ with $\tau+i\in\Omega$, define the
forward difference on the further-pinned residual space by
\[
 \Delta_i f(T):=f(T+i)-f(T),\qquad T\in\Omega^{\tau+i}.
\]
Both arguments belong to $\Omega^\tau$ by downward closure.
Writing a multilinear representative as $f=X_i g+h$, with $g,h$
independent of $X_i$, shows that $\Delta_i f=g$ on this space.
Thus a degree-$j$ function becomes a function of degree at most $j-1$.

The cost of taking this difference is
\begin{equation}\label{abs-eq:difference-L2}
 \|\Delta_i f\|_{2,\mu^{\tau+i}}
 \le\sqrt{\frac{1+\kappa}{b}}\,\|f\|_{2,\nu}.
\end{equation}
The two marginal assumptions play separate roles: $\kappa$ compares
the weights of the paired configurations $T$ and $T+i$, whereas $b$
controls the normalization when conditioning on occupation of $i$.
Neither step requires a lower bound on the probability of an individual
configuration. Specifically, weighted Cauchy--Schwarz gives, for every $T\in\Omega^{\tau+i}$,
\[
 \nu(T+i)|f(T+i)-f(T)|^2
 \le\left(1+\frac{\nu(T+i)}{\nu(T)}\right)
       [\nu(T+i)|f(T+i)|^2+\nu(T)|f(T)|^2].
\]
For fixed $i$, these pairs are disjoint. Sum and use the marginal-ratio
bound to obtain
\[
 \nu(i\in S)\|\Delta_i f\|_{2,\mu^{\tau+i}}^2
 \le(1+\kappa)\|f\|_{2,\nu}^2.
\]
Here $\mu^{\tau+i}(T)=\nu(T+i)/\nu(i\in S)$.
Dividing by the occupied-marginal lower bound $\nu(i\in S)\ge b$
proves \eqref{abs-eq:difference-L2}.

\paragraph{From local increments to the maximum norm.}
Let $a_j$ be the largest ratio $\|f\|_\infty/\|f\|_2$ for nonzero
functions of degree at most $j$, over all positive pinnings.
Then $a_0=1$. Choose $x\in\Omega^\tau$ with
$|f(x)|\le\|f\|_{2,\nu}$, which exists by averaging.
Choosing this starting point avoids paying for a possibly very small
probability of a prescribed configuration, such as the empty set.
For any $y\in\Omega^\tau$, delete $x\setminus y$ and then insert
$y\setminus x$. The path $x\to x\cap y\to y$ is legal by downward
closure and has at most $N$ steps. Each increment, up to sign, is a
value of some $\Delta_i f$ on $\Omega^{\tau+i}$. Therefore
\[
 |f(y)|\le\|f\|_{2,\nu}
       +N\max_i\|\Delta_i f\|_{\infty,\mu^{\tau+i}}
 \le\left(1+N\sqrt{\frac{1+\kappa}{b}}\,a_{j-1}\right)\|f\|_{2,\nu},
\]
where the last step uses the induction hypothesis and
\eqref{abs-eq:difference-L2}.
Singleton residual spaces need no path. Taking the supremum gives
\[
 a_j\le1+N\sqrt{\frac{1+\kappa}{b}}\,a_{j-1}
       \le\left(1+N\sqrt{\frac{1+\kappa}{b}}\right)^j,
\]
which proves the lemma.
\end{proof}

\smallskip\noindent\emph{Relation to marginal ratios under arbitrary pinning.}
The terminology \emph{bounded marginal ratios} follows
\cite[Definition 1.5]{CCCYZ25}. On downward-closed support,
\eqref{intro-eq:kappa} is equivalent to
\[
 \frac{\mu(X_i=1\mid X_\Lambda=\sigma)}
      {\mu(X_i=0\mid X_\Lambda=\sigma)}\le\kappa
\]
for every feasible pinning $\sigma\in\{0,1\}^\Lambda$ with
$\Lambda\subseteq[N]$ and $i\notin\Lambda$.
The denominators are positive by downward closure. One direction pins
all coordinates except $i$; the other averages full conditional
occupation probabilities, each at most $\kappa/(1+\kappa)$.

\section{Near-linear mixing for matchings}
\label{sec:sequential}
\label{sec:matching-mixing}\label{sec:profile}

We prove the near-linear mixing and sampling bounds of
\cref{thm:intro-sequential}: verify the spectral-stability and marginal
conditions of \cref{thm:intro-boolean}, obtain the mixing bound, and
implement each single-edge GD update in constant time.

Fix a simple graph $G=(V,E)$ with $m\ge1$ and activity $\lambda>0$,
and put $\mu=\mu_{G,\lambda}$. Matchings form a downward-closed
family with $N=m$ coordinates. The edgeless case simply returns the
empty matching. For $M\sim\mu$, define the dimer indicators $X_e=\one_{\{e\in M\}}$.
Write $\mu(e)=\E_\mu X_e$. 

\paragraph{Verifying spectral stability.}
Under a positive pinning $\tau$, the posterior $\nu_t^\tau$ is the
matching law on $H=G-V(\tau)$ at activity $\lambda(1-t)$.
The following covariance bound verifies field-dynamics spectral
stability of $\mu$.

\begin{proposition}[Edge covariance]
\label{prop:holant-edge}
For every finite simple graph at activity $s>0$, and after every
feasible edge pinning,
\begin{equation*}
 \Cov(X)\preceq\widehat C(s)\Diag(\E X),\qquad
 \widehat C(s)=
 \begin{cases}
 \min\{7+48s,(1-s)^{-1}\},&0<s<1,\\
 7+48s,&s\ge1.
 \end{cases}
\end{equation*}
Deterministic coordinates are omitted. Setting $\widehat C(0)=1$,
the function $\widehat C$ is nondecreasing and
$\widehat C(s)=1+O(s)$ as $s\downarrow0$.
\end{proposition}
\begin{proof}
For matchings, \cite[proof of Lemma 7]{ConcurrentHolant26} gives
$C_1=2+36s$ and $C_2=6$. The spectral argument in the proof of
Theorem 4 there then yields
\[
 \Cov(X)\preceq(2+36s+\sqrt{12})\Diag(\E X)
 \preceq(7+48s)\Diag(\E X).
\]
Combine with \cite[proof of Corollary 5.10]{CCCYZ25} for $s<1$.
Feasible edge pinning leaves a residual matching law, so both bounds apply.
\end{proof}
Thus $C(t)=\widehat C(\lambda(1-t))$ is an admissible
spectral-stability rate. Define
\begin{equation}\label{eq:c-lambda}
 I_\lambda=\int_0^\lambda\frac{\widehat C(s)-1}{s}\,ds,
 \qquad c_\lambda=e^{-I_\lambda}>0.
\end{equation}
Note that $I=I_\lambda<\infty$ in
\eqref{eq:integrated-stability}.
Moreover, $I_\lambda$ is nondecreasing and graph-independent, so
$c_\lambda^{-1}$ is uniformly bounded over $0<\lambda\le\lambda_{\max}$
for every fixed finite $\lambda_{\max}$.

\paragraph{Verifying the marginal conditions of $(\kappa, b)$.}

Every legal insertion has weight ratio $\lambda$, so bounded marginal
ratios hold with $\kappa=\max\{1,\lambda\}$.

For occupied-marginal lower bounds, let $H$ be any positively pinned
residual graph and $\Delta$ the maximum degree of $G$.
For $uv\in E(H)$, the one-vertex partition-function decomposition gives
\begin{align*}
 Z_H(\lambda)
 &=Z_{H-u}(\lambda)+\lambda\sum_{x\sim_H u}Z_{H-u-x}(\lambda)
 \le(1+\lambda d_H(u))Z_{H-u}(\lambda)\\
&\le(1+\lambda d_H(u))(1+\lambda d_{H-u}(v))Z_{H-u-v}(\lambda).
\end{align*}
Since $\mu_{H,\lambda}(uv)=\lambda Z_{H-u-v}(\lambda)/Z_H(\lambda)$,
\begin{equation*}
 \mu_{H,\lambda}(uv)
 \ge\frac{\lambda}
 {(1+\lambda d_H(u))(1+\lambda d_{H-u}(v))}
 \ge\frac{\lambda}{(1+\lambda\Delta)^2}.
\end{equation*}
Taking the last quantity as $b$ verifies \eqref{intro-eq:b}.
Degree enters only through this marginal bound, whose inverse costs
only a logarithmic factor in the criterion.

\paragraph{Poincar\'e inequality.}
We record the relaxation-time bound for the single-edge GD separately from the LSI consequence below.
This bound is used for approximate counting in Section~\ref{sec:preconditioned-js}.
\begin{proposition}[Poincar\'e inequality {\cite[Corollary 5 and Theorem 10]{ConcurrentHolant26}}]
\label{prop:matching-gap}
For every $\lambda>0$,
\begin{equation*}
 \gap(P_{\mathrm{GD}})^{-1}\le m(1+\lambda)c_\lambda^{-1}=O_\lambda(m).
\end{equation*}
The same bound holds on nontrivial residual matching instances.
\end{proposition}
\begin{proof}
Apply \cite[Theorem 10, Eq.~(2.2)]{ConcurrentHolant26} with the rate
$C(t)$ verified above and insertion ratio $\lambda$.
\end{proof}

\paragraph{Log--Sobolev inequality and mixing.}
The logarithmic factor from the marginal bound is controlled by
\begin{equation}\label{eq:A-lambda}
 A_\lambda=
 2\log\!\left(1+m(1+\lambda(n-1))
                      \sqrt{\frac{1+\lambda}{\lambda}}\right)
 =O\bigl(\log n+|\log\lambda|\bigr).
\end{equation}
Here $A_\lambda\ge1$. We obtain the following bounds with a universal
constant $C$.

\begin{corollary}[Log--Sobolev inequalities and mixing for matchings]
\label{thm:main}
For every $\lambda>0$, single-edge GD satisfies, with the LSI and MLSI conventions
\eqref{eq:lsi}--\eqref{eq:mlsi},
\begin{equation}\label{eq:matching-LS}
 \alpha_{\rm LS}^{-1},\ \alpha_{\rm mLS}^{-1}
 \le C m(1+\lambda)c_\lambda^{-1}A_\lambda.
\end{equation}
Writing $\mu_{\min}=\min_M\mu(M)$, for $0<\varepsilon\le1/2$,
\begin{equation}\label{eq:discrete-time}
 t_{\rm mix}(P_{\mathrm{GD}},\varepsilon)
 \le C m(1+\lambda)c_\lambda^{-1}
 \left[
 A_\lambda\log\log\frac{e^2}{\mu_{\min}}
 +\log(1/\varepsilon)
 \right].
\end{equation}
These bounds hold on residual graphs with at least one edge
after feasible edge pinnings; an edgeless residual law is a point mass.
At fixed activity, the inverse constants are $O_\lambda(m\log n)$
and the mixing bound is
$O_\lambda(m[\log^2 n+\log(1/\varepsilon)])$.
\end{corollary}

\begin{proof}
The verified parameters are
\[
 N=m,\qquad \kappa=\max\{1,\lambda\},\qquad
 b=\frac{\lambda}{(1+\lambda\Delta)^2},\qquad I=I_\lambda.
\]
Since $\kappa\le1+\lambda$ and $\log(m\kappa/b)\le A_\lambda$,
\cref{thm:intro-boolean} gives the LSI bound in
\eqref{eq:matching-LS} and the mixing bound \eqref{eq:discrete-time};
the MLSI bound follows from \eqref{eq:mlsi} and
$(a-b)(\log a-\log b)\ge4(\sqrt a-\sqrt b)^2$ for $a,b>0$.

Since $Z_G(\lambda)\le(1+\lambda)^m$ and every matching has weight
at least $\min\{1,\lambda^m\}$,
\begin{equation}\label{eq:matching-min-mass}
 \log\frac1{\mu_{\min}}
 \le m\log\bigl(1+\max\{\lambda,\lambda^{-1}\}\bigr).
\end{equation}
At fixed $\lambda$, this is $O_\lambda(m)$; using $m\le n^2/2$ and
$A_\lambda=O_\lambda(\log n)$ gives the claimed mixing time.
Feasible edge pinnings leave vertex- and edge-deleted simple graphs,
so the same applies.
\end{proof}

\paragraph{Near-linear sampling time.}
Starting from the empty matching and storing each vertex's matched
neighbor (or an unmatched flag) gives $O(1)$ work per GD update and
total sampling work $O(m+n+t_{\rm mix})$, proving
\cref{thm:intro-sequential}.

\section{Work-efficient parallel simulation of Glauber dynamics}
\label{sec:parallel-simulation}\label{sec:parallel-algorithm}

We prove \cref{thm:intro-parallel} by evaluating the same GD updates
in parallel. Section~\ref{sec:exact-parallel} establishes correctness and
the dependency structure, Section~\ref{sec:parallel-rounds} bounds the
number of parallel rounds, and Section~\ref{sec:parallel-complexity}
implements them with near-linear work.

Put $p=\lambda/(1+\lambda)$ and let $\Delta$ be the maximum degree;
assume $m\ge1$.

\paragraph{Parallel simulation algorithm.}
Given an initial matching and an integer $N\ge1$, generate $N$ \emph{update instructions}
$(e_s,B_s)$ with mutually independent uniform edge choices $e_s$ and
Bernoulli$(p)$ coins $B_s$. Instruction $s$ deletes $e_s$ if present,
then inserts it if $B_s=1$ and both endpoints are unmatched.
Process $B=\lceil N/m\rceil$ consecutive batches of at most
$m$ instructions, using each output matching as the next input.
Within each batch, index its $T$ updates by integers $s=1,\ldots,T$;
time $t\in\{0,\ldots,T\}$ denotes the state after $t$ updates.
\begin{enumerate}[label=(\arabic*)]
\item \emph{Records.} A record represents a proposed occupation interval;
accepting it means that the edge is occupied throughout that interval.
Sort instructions by edge and time. Each $B_s=1$ creates
$(e_s,[s,d))$, where
$d=\min\bigl(\{t\in\{s+1,\ldots,T\}:e_t=e_s\}\cup\{T+1\}\bigr)$.
Each edge $e$ of the batch's input matching creates $(e,[0,d))$, where
$d=\min\bigl(\{t\in\{1,\ldots,T\}:e_t=e\}\cup\{T+1\}\bigr)$.
Coin-zero updates create no records but still end previous records of
their edge. Thus records of the same edge have disjoint intervals.
\item \emph{Decisions.} Let $\mathcal R$ be the set of all records and
$\mathcal R_0\subseteq\mathcal R$ those starting at time zero. Write $x\sim y$ if the
records \emph{conflict}: their edges share an endpoint and their intervals
overlap. Write $t_x$ for the start time of $x$.
Accept every record in $\mathcal R_0$, and let
\[
 \mathcal C_{\rm init}:=\{x\in\mathcal R\setminus\mathcal R_0:
        x\sim y\text{ for some }y\in\mathcal R_0\},
 \qquad
 U_0:=\mathcal R\setminus(\mathcal R_0\cup\mathcal C_{\rm init}).
\]
Reject the records in $\mathcal C_{\rm init}$. For $k\ge0$, while the
\emph{active} (unresolved) set $U_k$ is nonempty, define the \emph{ready} records by
\[
 \mathcal Q_k:=\{x\in U_k:
       \text{there is no }y\in U_k\text{ with }t_y<t_x\text{ and }y\sim x\},
\]
and
\[
 \mathcal C_k:=\{y\in U_k\setminus\mathcal Q_k:
       y\sim x\text{ for some }x\in\mathcal Q_k\},
 \qquad
 U_{k+1}:=U_k\setminus(\mathcal Q_k\cup\mathcal C_k).
\]
Simultaneously accept $\mathcal Q_k$ and reject $\mathcal C_k$.
Readiness and distinct noninitial start times ensure that records in
$\mathcal Q_k$ do not conflict. These iterations are the \emph{peeling rounds},
not GD time steps.
\item \emph{Output.} Return the edges of all accepted records whose intervals contain $T$.
\end{enumerate}

\begin{theorem}[Exact discrete simulation]\label{thm:simulation}
For $m\ge1$, an integer $N\ge1$, and $0<\eta\le1/4$, the batch
algorithm exactly simulates $N$ random-scan single-edge GD updates from any
given matching. Let $r=\nu(G)$ be the maximum matching size of $G$
and put $S=(m+n+N+10)/\eta$. Its work is
$(m+n+N)\log^{O(1)}S$ on every run, and with probability at
least $1-\eta$ its depth is
\begin{equation}\label{eq:simulation-depth}
 \left(1+\frac Nm\right)
 \left(1+\min\{p\Delta,m^{1/3},\sqrt{p r}\}\right)\log^{O(1)}S.
\end{equation}
Here random edge choices and Bernoulli coins are exact. The algorithm
always runs to completion; $\eta$ bounds its depth tail, not its
sampling error.
\end{theorem}


\subsection{Correctness and dependency structure}\label{sec:exact-parallel}

Fix a batch of $1\le T\le m$ updates. We first prove that peeling reproduces
GD, then control successive accepted occupations and the rejected updates
that link them. The batch length ensures at most one expected update per
edge, which limits the propagation of dependencies.

\begin{lemma}[Interval representation]\label{lem:interval-exact}
Greedily accepting records in increasing start order, subject to having
no conflict with previously accepted records, reproduces every state
in the batch. Its terminal matching consists of the accepted records
whose intervals contain $T$.
\end{lemma}
\begin{proof}
At update $s$, any previous record of $e_s$ expires. The edge is occupied
after the update precisely if $B_s=1$ and neither endpoint is covered
at time $s$ by an earlier accepted record. Induction on $s$ gives the GD trajectory.
Parallel peeling agrees with this greedy rule: ready records are
pairwise nonconflicting and must be accepted; removing them and their
conflicts preserves all remaining decisions.
\end{proof}

\paragraph{Padding for conditional independence.}
We will reveal accepted occupations while retaining control of rejected
updates. For this analysis only, generate $2T$ independent slots, each
with probability $1/2$ performing no update (a \emph{no-op}) and otherwise
performing an ordinary GD update. The no-op option keeps conditional
probabilities controlled after accepted occupations are revealed
(\cref{lem:visible}). Conditioning on exactly $T$ ordinary updates recovers
the original batch with uniformly placed no-ops. Since
\begin{equation*}
 \Prb\{\operatorname{Bin}(2T,1/2)=T\}\ge(2T+1)^{-1},
\end{equation*}
an unconditioned failure probability $\zeta:=\eta/[B(2T+1)]$ becomes
at most $\eta/B$ after conditioning. No-ops change neither conflicts nor
peeling rounds and are never generated by the algorithm. Until the
assembly in Section~\ref{sec:parallel-complexity}, probabilities refer
to the \emph{unconditioned} padded experiment, and record times are slot
indices $0,\ldots,2T$, with terminal endpoint $2T+1$.

\paragraph{Controlling successive accepted occupations.}
Let $\mathcal A$ be the accepted records, including initial ones, and
let $F$ have vertex set $\mathcal A$. At each vertex of $G$, direct an
arc from each incident accepted record to the next one chronologically,
merging repeated arcs. This gives a directed acyclic graph with in- and
out-degree at most two. Write $\mathcal D(x)$ for the records reachable
from $x$, including $x$ itself.

Recall $r=\nu(G)$ and $S=(m+n+N+10)/\eta$.
\begin{lemma}[Descendant bound]\label{lem:descendants}
Conditionally on the slot outcomes through an accepted record $x$'s
start time $t_x$ (time zero for initial records), for every integer $k\ge1$,
\begin{equation}\label{eq:descendant-tail}
 \Prb\{|\mathcal D(x)|\ge k\}\le C e^{-ck}
\end{equation}
for absolute $c,C>0$. Consequently, for $L=\lceil C\log S\rceil$
with a sufficiently large absolute constant, except with probability
at most $\zeta/4$, every accepted record has at most $L$ descendants, every
vertex of $G$ is incident to at most $L$ accepted records, and
\begin{equation}\label{eq:accepted-count}
 h:=|\mathcal A|\le2rL.
\end{equation}
\end{lemma}
\begin{proof}
Explore descendants forward from $t_x$, initially tracking only $x$
with no marks. A tracked record is \emph{live} while its interval is current.
When its edge is updated, replace it by a waiting mark at each endpoint.
The next accepted record there consumes the mark and becomes live;
marks reaching the same record merge. An update of an occupied edge
with coin one immediately starts a new record, a \emph{renewal}, and
is included in this rule.

Let $U$ count updates to tracked live edges. Each such update replaces
one live record by at most two marks. Thus there are at most $1+U$
live records and marks combined, and at most $1+2U$ records have been
discovered. The next slot updates a tracked live edge with probability
at most $(1+U)/(2m)$. Consequently, $1+U$ is stochastically dominated
by the process $Y_t$, $0\le t\le2T$, starting at $Y_0=1$, which increments
from $y$ to $y+1$ with probability $y/(2m)$ and otherwise stays at $y$.
There are at most $2T\le2m$ remaining slots, so
$|\mathcal D(x)|$ is stochastically dominated by $2Y_{2T}-1$.
Write $\rise{x}{a}=x(x+1)\cdots(x+a-1)$ for the rising factorial.
For every integer $a\ge1$,
\[
 \E[\rise{Y_{t+1}}a\mid Y_t]
 =\rise{Y_t}a\left(1+\frac a{2m}\right),\qquad
 \E\rise{Y_{2T}}a\le a!e^a.
\]
Thus $\Prb(Y_{2T}\ge k)\le(ea/k)^a$; taking
$a=\lfloor k/(2e)\rfloor$ gives an exponential tail. The domination above
then proves \eqref{eq:descendant-tail}, adjusting constants for small $k$.
Union-bound over the at most $2T+r\le3m$ possible starting labels
(slots and initial records), conditioning only on the past at each start.
Records incident to one vertex of $G$ form an $F$-path and hence number
at most $L$. The endpoints of a maximum matching form a vertex cover
of size $2r$; charging each record to one of these endpoints gives
\eqref{eq:accepted-count}.
\end{proof}

\begin{lemma}[Conditional rejected-update bound]\label{lem:visible}
Let $\mathcal F_{\rm vis}$ be the information generated by the matching
after every padded slot and the slot and edge identities of all renewals.
It determines the accepted records, their start and end times, and $F$.
Call a slot \emph{unrevealed} if it is neither a state-changing update
nor a renewal. Conditional on $\mathcal F_{\rm vis}$, these slots have
independent outcomes, and for every edge $e$ and unrevealed slot,
\begin{equation*}
 \Prb\{\text{rejected coin-one update of }e\mid\mathcal F_{\rm vis}\}
 \le p/m.
\end{equation*}
\end{lemma}
\begin{proof}
State changes identify their update edge and coin. Together with the
renewals, they specify when every accepted occupation begins and ends.
Fix this visible information. Each unrevealed slot allows exactly a no-op,
a coin-zero update on an absent edge, or a coin-one update on an absent
edge blocked by an occupied incident edge. These are separate restrictions
on independent slots, so conditional independence is preserved.
Each allowed set has probability at least $1/2$, since it includes the
no-op. A specified rejected update therefore has conditional probability
at most $(p/(2m))/(1/2)=p/m$.
\end{proof}

\subsection{Bounding the number of peeling rounds}
\label{sec:parallel-rounds}

A long execution requires a long conflict chain. The descendant bound
limits its successive accepted occupations; it remains to count the
rejected \emph{connectors} linking these short pieces.

If peeling takes $R$ rounds, backward tracing gives accepted records
$x_0,\ldots,x_{R-1}$ and rejected records $z_0,\ldots,z_{R-2}$ with
\begin{equation}\label{eq:interval-witness}
 t_{x_0}<t_{z_0}<t_{x_1}<\cdots<t_{z_{R-2}}<t_{x_{R-1}},
\end{equation}
where $z_i$ conflicts with both $x_i$ and $x_{i+1}$.
Trace backwards: a record selected in round $j$ had an earlier blocker
removed in round $j-1$ by an earlier acceptance. Initial records were
processed separately, so these birth times are positive and distinct.

If $x_i,x_{i+1}$ share a physical endpoint, then
$x_{i+1}\in\mathcal D(x_i)$. Otherwise $z_i$ uses one of at most four
edges joining their endpoint pairs: a \emph{connector step}.
With $q$ connector steps, the $q+1$ intervening successor pieces each
have at most $L$ records by Lemma~\ref{lem:descendants}, so
\begin{equation}\label{eq:contracted-length}
 R\le L(q+1).
\end{equation}

Let $\mathcal W_q$ be the event that a witness has at least $q$ connectors.
Fix visible data satisfying the descendant bound. Chronological
connector pairs $(u_i,v_i)$ have disjoint open time intervals and
$u_{i+1}\in\mathcal D(v_i)$, with equality allowed.
For fixed pairs, discard the overlap requirements and ask only for a
rejected connector update in each interval. These events use disjoint
unrevealed slots, even when a physical connector is reused.
Lemma~\ref{lem:visible} and AM--GM bound their joint probability by
\begin{equation}\label{eq:connector-product}
 \prod_{i=1}^q\frac{4p}{m}(t_{v_i}-t_{u_i})
 \le\left(\frac{8pT}{mq}\right)^q\le(8p/q)^q.
\end{equation}
\paragraph{Degree and matching size.}
Each source has at most $2\Delta L$ targets: its endpoints have at most
$2\Delta$ neighboring vertices, each in at most $L$ accepted records.
Alternatively, choose the increasing target list in $\binom hq$ ways.
There are $h$ choices for $u_1$ and at most $L$ for each subsequent source.
Thus, for $1\le q\le h$,
\begin{equation*}
\begin{split}
 \Prb(\mathcal W_q\mid\mathcal F_{\rm vis})
 &\le hL^{q-1}\min\{(2\Delta L)^q,\tbinom hq\}(8p/q)^q\\
 &\le\frac hL\min\left\{
    \left(\frac{16p\Delta L^2}{q}\right)^q,
    \left(\frac{8epLh}{q^2}\right)^q\right\}.
\end{split}
\end{equation*}
Truncate longer witnesses; for $q>h$ the event is empty.
Take $q=\lceil C[1+\log(S/\zeta)+
\min\{p\Delta L^2,\sqrt{pLh}\}]\rceil$.
Since $h\le2rL$ and $R\le L(q+1)$,
\begin{equation}\label{eq:degree-size-rounds}
 R=O\bigl((1+\min\{p\Delta,\sqrt{pr}\})\log^3 S\bigr)
\end{equation}
except with probability at most $\zeta/2$, including descendant failure.
The matching-size term alone needs only $\log^2 S$.

\paragraph{Using the number of edges.}
For the cube-root bound, count physical connectors by contracting
successor paths. The resulting ancestor choices are controlled by
their total path length, not a separate factor $L$ at each step.

\begin{lemma}[Total successor-path length]\label{lem:several-paths}
Let $\mathcal B_q$ be the event that $F$ has $q$ directed paths whose
total number of arcs is at least $3q$, with each path's final birth time
strictly before the next path's initial birth time. Paths with no arcs
are allowed. Before conditioning on $\mathcal F_{\rm vis}$,
\begin{equation}\label{eq:several-paths}
 \Prb(\mathcal B_q)\le(Cm/q^3)^q.
\end{equation}
\end{lemma}
\begin{proof}
Truncate to total length $3q$. Specify $q$ root birth labels (slots or
initial edges), lengths $\ell_1+\cdots+\ell_q=3q$, and an endpoint
at each step: at most
\[
 (3m)^q\binom{4q-1}{q-1}2^{3q}\le(3m)^q2^{7q}
\]
certificates. Follow each causally, failing at an unaccepted root or
if the previous path has not finished strictly earlier. Starting after
the root's birth, each successor step waits for its occupied edge's
update and then acceptance at the chosen endpoint.
Only one edge is tracked at a time, with conditional hit probability
at most $1/(2m)$. The required $3q$ distinct hits have probability at most
\[
 \binom{2T}{3q}(2m)^{-3q}\le(e/(3q))^{3q}.
\]
This sums over hit positions without future conditioning. Multiply by
the certificate count; if $3q>2T$, the event is empty.
\end{proof}

Contract each connector and its following successor path into one step:
\[
 x\xrightarrow{\text{connector}}y
  \xrightarrow{F\text{-path}}z.
\]
Form a simple graph $J$ on noninitial accepted records, joining $x$ to $z$ whenever
some $y$ satisfies $t_x<t_y\le t_z$, $z\in\mathcal D(y)$, and the
disjoint physical edges of $x,y$ have a connector. Dropping the initial
successor prefix maps every witness to an increasing $J$-path with one
edge per connector. On the descendant event, each physical connector
gives at most $2L^2$ pairs $(x,y)$ and $L$ choices of $z$ per pair; hence
\begin{equation*}
 |E(J)|\le M:=2mL^3.
\end{equation*}

Fix visible data satisfying the descendant event and $\mathcal B_q^c$;
both are $\mathcal F_{\rm vis}$-measurable, so \cref{lem:visible} still
applies. For an increasing $J$-path $x_0,\ldots,x_q$, let $d_i$ be the
maximum length of an $F$-path ending at $x_i$ and starting in
$(t_{x_{i-1}},t_{x_i}]$, allowing singletons. These paths are strictly
time-separated, so $\sum_i d_i<3q$. In-degree at most two gives at most
$2^{d_i+1}-1$ ancestors $y$ in the $i$th interval and
\begin{equation*}
 \prod_{i=1}^q2^{d_i+1}\le2^{4q}.
\end{equation*}
Each step requires a rejected update on one of at most four connectors per
ancestor. As in \eqref{eq:connector-product}, the necessary events use
disjoint open slot intervals. Their joint conditional probability is at most
\[
 \prod_{i=1}^q\frac{4p}{m}\,2^{d_i+1}
       (t_{x_i}-t_{x_{i-1}})
 \le (128p)^q\prod_{i=1}^q(\tau_{x_i}-\tau_{x_{i-1}}),
 \qquad \tau_x=t_x/(2T).
\]
The next estimate sums these gap products over $J$-paths.

\begin{lemma}[Weighted increasing paths]\label{lem:weighted-paths}
Let $J$ be a simple graph with at most $M$ edges and distinct vertex
times $\tau_x\in[0,1]$. For $q\ge1$ and $k=\lceil(q+1)/2\rceil$,
\begin{equation}\label{eq:weighted-paths}
 \sum_{\substack{\tau_{x_0}<\cdots<\tau_{x_q}\\
                  x_{i-1}x_i\in E(J)}}
 \prod_{i=1}^q(\tau_{x_i}-\tau_{x_{i-1}})
 \le\frac1{q!}(M/k)^k.
\end{equation}
\end{lemma}
\begin{proof}
Write each gap as an integral over a cut between its endpoint times.
The $q$ cuts lie in the ordered simplex of volume $1/q!$ and divide
the vertices into consecutive layers. Let $m_i$ count edges between
layers $i-1$ and $i$; these edge sets are disjoint, so $\sum_i m_i\le M$.
When $q$ is odd, path edges in positions $1,3,\ldots,q$ determine all
path vertices. When $q$ is even, use $1,3,\ldots,q-1,q$ instead.
These are $k$ positions, so the path count for fixed cuts is bounded
by a product of $k$ of the $m_i$, at most $(M/k)^k$ by AM--GM.
Integration proves \eqref{eq:weighted-paths}.
\end{proof}

Applying \cref{lem:weighted-paths} and $q!\ge(q/e)^q$, with
$k=\lceil(q+1)/2\rceil$, gives, for $1\le q\le M$,
\begin{equation}\label{eq:cube-tail}
 \Prb(\mathcal W_q\mid\mathcal F_{\rm vis})
 \le\frac{(128p)^q}{q!}(M/k)^k
 \le(M+1)\left(\frac{Cp\sqrt M}{q^{3/2}}\right)^q.
\end{equation}
For $q>M$ no such path exists. Choose
\[
 q=\left\lceil C\left[m^{1/3}+(p^2M)^{1/3}
                  +1+\log\frac{M+2}{\zeta}\right]\right\rceil.
\]
The first two terms control \eqref{eq:several-paths} and
\eqref{eq:cube-tail}; the logarithm reduces their combined failure to
$3\zeta/4$. Including descendant failure and using $M=2mL^3$,
$p\le1$, and \eqref{eq:contracted-length} gives
\begin{equation}\label{eq:cube-rounds}
 R=O\bigl(m^{1/3}\log^2 S\bigr)
\end{equation}
except with probability at most $\zeta$. The $m^{1/3}$ term is independent of $p$;
this argument does not give a $p^{2/3}$ improvement at small activity.

\subsection{Near-linear work and the sampling guarantee}
\label{sec:parallel-complexity}\label{sec:interval-implementation}

Per-vertex interval trees find ready records and remove conflicts
without constructing the conflict graph. Each record
incurs only polylogarithmic work, regardless of the number of rounds.

\begin{lemma}[Cost of interval peeling]\label{lem:interval-implementation}
For a batch of $T\ge1$ update instructions and an initial matching $M_0$, the
interval-peeling algorithm has a CREW implementation with total work
\[
 (T+|M_0|)\log^{O(1)}(T+|M_0|+2)
\]
on every instruction sequence. Preparing the records, executing each peeling
round, and extracting the terminal matching each take polylogarithmic
depth. After reading the graph once, the implementation uses only the
edges specified by the instructions and the current matching list.
\end{lemma}

\begin{proof}
For $P\le T+|M_0|$ records, store both endpoint incidences sorted by
vertex and start time. Trees use these static arrays only at touched
vertices. If $P=0$, peeling has constant control cost.

\paragraph{Finding ready records without rescanning.}
For an active incidence $I=[s,d)$ at vertex $v$, maintain its
earlier-blocker count
\[
 c_v(I)=\#\{\text{active incidences }J=[s_J,d_J)
                        \text{ at }v:s_J<s<d_J\}.
\]
Sorting starts and ends and taking prefix sums gives initial counts.
Removing $J$ decrements counts with starts in $(s_J,d_J)$.
A segment tree with lazy range-add tags and subtree minima supports
these updates, deletion, and zero reporting. Report zeros after removing
initial records and their conflicts, then only newly attained zeros.

Give reported incidences value $10P$ and a separate ready flag; give
removed incidences the same sentinel and mark them inactive.
At most $P$ further decrements remain, so each incidence is reported
at most once. It continues blocking until actual removal, which
decrements its original interval. Group reports by record, update
endpoint flags, and enqueue each active record with both flags ready
once, excluding records removed in that batch.

For CREW execution, group range updates by their $O(\log(P+1))$
canonical nodes, sum shared-node updates, and recompute touched
ancestors levelwise. Batch deletions likewise; report newly attained zeros
only in updated trees, descending only zero-minimum subtrees.
This costs polylogarithmic depth per round and work per operation,
with $O(P)$ updates, removals, and reports.

\paragraph{Removing conflicts without repeated reporting.}
A second tree stores subtree maximum active ends. Canonical range
decomposition and maximum-based pruning report intervals in a start
range whose ends exceed a threshold, with polylogarithmic depth and
polylogarithmic work per query and reported interval.

At a fixed vertex, sort the disjoint intervals accepted in a round as
$[s_1,d_1),\ldots,[s_k,d_k)$. Query starts in $[d_{j-1},d_j)$ and
ends greater than $s_j$, taking $d_0=-\infty$.
A record first intersecting interval $j$ starts below $d_j$ but not
below $d_{j-1}$, which would force an earlier intersection.
Thus all conflicts are reported, at most once per endpoint because
the start ranges are disjoint. Sort to deduplicate endpoint reports;
remove accepted records as well.

Process initial records by the same queries before initializing ready
flags. Each query belongs to a record accepted once, giving
$\widetilde O(P)$ total work and polylogarithmic depth per round.
Sorting instructions for next updates, initializing trees, and extracting
compact matching lists also fit the claimed work and depth.
\end{proof}

\begin{proof}[Proof of Theorem~\ref{thm:simulation}]
Correctness follows from \cref{lem:interval-exact} and its proof.
For batch lengths $T_j$ and initial matchings $M_{j-1}$,
\[
 \sum_{j=1}^B(T_j+|M_{j-1}|)\le N+Br\le2N+m,
\]
since $|M_{j-1}|\le r\le m$. Lemma~\ref{lem:interval-implementation}
and the one-time $O(m+n)$ graph initialization give the work bound on every run.
This accounting charges work to records, not to rounds: each record
is removed once, and each endpoint incidence is reported ready at most
once. The data structures avoid scanning unresolved records in every
round. Thus a long dependency chain increases depth without multiplying
the total work by the number of rounds.

Take the smaller bound in \eqref{eq:degree-size-rounds} and
\eqref{eq:cube-rounds}. Undoing the padding conditioning gives failure
at most $(2T+1)\zeta=\eta/B$ per batch, uniformly over its initial matching.
Although a batch's input matching depends on preceding batches,
conditioning on their history fixes this input while leaving the batch's
instructions independent with their original law. Its conditional
failure probability is therefore still at most $\eta/B$; independence
between batch failures is not needed.
A union bound over $B\le1+N/m$ batches and polylogarithmic depth per
round give \eqref{eq:simulation-depth}. Only the analysis selects a
bound; the algorithm computes neither $F,J$ nor a maximum matching.
\end{proof}

\begin{proof}[From simulation to sampling: proof of Theorem~\ref{thm:intro-parallel}]
Choose $N=O_\lambda(m[\log^2 n+\log(1/\varepsilon)])$ from
\eqref{eq:discrete-time} using \eqref{eq:matching-min-mass}.
Theorem~\ref{thm:simulation} preserves its GD output law, hence its
$\varepsilon$ sampling error. Since $m,r,\Delta\ge1$ and $p\le1$,
\[
 1+\min\{p\Delta,m^{1/3},\sqrt{pr}\}
 \le 2\min\{\Delta,m^{1/3},\sqrt r\}.
\]
Since $r\le n/2$ and $m\le n^2/2$, substitution in
\cref{thm:simulation} gives the claimed work and depth, with the
remaining factors polylogarithmic in $n$, $1/\varepsilon$, and $1/\eta$.
\end{proof}

If a deterministic depth cap is desired, replace $r=\nu(G)$ by $n/2$ in
\eqref{eq:simulation-depth}, choose the tail budget at most
$\varepsilon/2$, and return the empty matching on overflow. This
can be coupled to the uncapped run using the same update instructions:
their outputs agree unless overflow occurs, so their total-variation
distance is at most the overflow probability. Running GD to
error $\varepsilon/2$ yields the same near-linear
work and $\widetilde O_\lambda(\min\{\Delta,m^{1/3},\sqrt n\})$ depth
on every run.

\section{Faster counting via learned Jerrum--Sinclair dynamics}
\label{sec:preconditioned-js}

The counting reduction needs repeated samples at nearby activities.
We first prove it from warm-start sampling
(Section~\ref{sec:counting}), then construct the learned
Jerrum--Sinclair chain and its monomer weights
(Section~\ref{js:sec:learned-kernel}).
Section~\ref{js:sec:learned-mixing} derives warm-start mixing from a star
relaxation bound by two Dirichlet-form comparisons.
Section~\ref{sec:star-inequalities} proves that bound, with auxiliary
gap estimates.

Throughout, $\mu=\mu_{G,\lambda}$ on a simple graph with $n\ge2$ and
$\lambda>0$. Recall the dimer indicators $X_e$ from
Section~\ref{sec:sequential} and set $J_v=1-\sum_{e\sim v}X_e$.
Write $v^0=\{J_v=1\}$,
so $\mu(v^0)=\E_\mu J_v$, and set $H=|M|$.
Concatenated events denote intersections.
For every fixed finite $\lambda_{\max}>0$, all analytic
constants below can be chosen uniformly over $0<\lambda\le\lambda_{\max}$.

\subsection{Counting from warm-start sampling}\label{sec:counting}

The connection to counting is the identity
\begin{equation*}
 \log Z_G(\lambda)-\log Z_G(\lambda_0)
 =\int_{\log\lambda_0}^{\log\lambda}
                         \E_{\mu_{G,e^\beta}}H\,d\beta.
\end{equation*}
Choose $\lambda_0$ so that $Z_G(\lambda_0)$ is close to one, discretize
this integral, and estimate its integrand by matching cardinalities
along an annealing trajectory. The identity
$d\E H/d\log\lambda=\Var H\le n/8$
(\cref{cnt:lem:cov}) controls both the quadrature error and the
single-sample variance. The remaining issue is to control correlations
between observations at different activities.

Here is the sampling guarantee used by the reduction. With fixed
weights $\mu(v^0)\le w_v\le B\mu(v^0)$, $w_v\le1$, the learned kernel
$\widehat P_w$ defined in \cref{js:sec:learned-kernel} satisfies
\[
 \nu\le W\mu
 \quad\Longrightarrow\quad
 \TV(\nu\widehat P_w^N,\mu)\le\delta
 \quad\text{for }N=\widetilde O_{\lambda,B}(n),
\]
where the suppressed factors are logarithmic in $n,W,1/\delta$.
Each step takes $O(\log n)$ work. The mixing statement is
\cref{js:lem:learned-mixing}, proved from a WPI below; the weights
are learned in $\widetilde O_\lambda(m+n)$ work by
\cref{js:lem:learning}.

Nearby activities make the preceding stationary law a warm start for
the next block. To bound correlations, we also apply the same guarantee
to the law with density $(H+1)/(\E H+1)$ relative to that stationary
law. Both reference inputs are at most $(n+2)$-warm at the next
activity. This is why warm-start mixing suffices for counting even
though successive matchings are dependent.

\paragraph{Counting algorithm.}
Given target $\lambda>0$ and $0<\varepsilon\le1/2$, return one if
$m\lambda\le\varepsilon/100$; this includes every edgeless graph.
Otherwise use the integers $b,s,N$ specified in
\cref{cnt:sec:algorithm}, and put $\lambda_0=\lambda2^{-b}$ and $K=bs$.
Their scales are $\lambda_0=\Theta(\varepsilon/m)$,
$K=\widetilde O_\lambda(n/\varepsilon^2)$, and
$N=\widetilde O_\lambda(n)$.
\begin{enumerate}[label=(\arabic*)]
\item Divide each of the $b$ doubling intervals from $\lambda_0$ to
$\lambda$ into $s$ equal linear pieces, obtaining endpoints
$\lambda_0<\cdots<\lambda_K=\lambda$.
\item Independently prelearn one guide table at each doubling interval's
lower endpoint by \cref{js:lem:learning}, with failure probability
$1/(100b)$ per table and randomness independent of annealing.
\item Start empty. For $i=0,\ldots,K$, run
$Q_i=\widehat P_i^N$ and record $H_i=|M_i|$, carrying the matching
between blocks. Here $\widehat P_i$ is the learned kernel at $\lambda_i$,
with stationary law $\mu_i=\mu_{G,\lambda_i}$, using the table for the
interval containing $(\lambda_{i-1},\lambda_i]$,
or the first table when $i=0$.
\end{enumerate}
Return
\begin{equation}\label{cnt:eq:estimator}
 \widehat Z=\prod_{i=1}^K(\lambda_i/\lambda_{i-1})^{H_i}.
\end{equation}
This thermodynamic-integration estimator (path sampling)
\cite{GelmanMeng98} exponentiates a quadrature of the mean cardinality;
it is not an unbiased telescoping product of ratio estimates.

We now give the parameters and verify the counting reduction, using
the sampling and learning guarantees stated above.

\subsubsection{Parameter choices}\label{cnt:sec:algorithm}

The direct-return branch is accurate because
\[
 1\le Z_G(\lambda)\le(1+\lambda)^m\le e^{m\lambda},
 \qquad 1-Z_G(\lambda)^{-1}\le m\lambda\le\varepsilon/100.
\]
On the nontrivial branch set
\[
 b=\left\lceil\log_2\frac{100m\lambda}{\varepsilon}\right\rceil,
 \qquad s=\left\lceil\frac{2^{16}nb}{\varepsilon^2}\right\rceil.
\]
Then $b\ge1$ and
\begin{equation}\label{cnt:eq:starting-activity}
 \frac{\varepsilon}{200m}<\lambda_0\le\frac{\varepsilon}{100m}.
\end{equation}
The schedule satisfies
\[
 a_i=\log(\lambda_i/\lambda_{i-1})\le1/s,\quad
 L=\sum_i a_i=b\log2,\quad A_2=\sum_i a_i^2\le b/s.
\]
Choose
\begin{equation*}
 N=\left\lceil C_\lambda n\left(1+\left\lceil
             \log_2\frac{n}{\delta}\right\rceil\right)^5\right\rceil,
 \qquad
 \delta=\frac{\varepsilon^2}
 {2^{20}(n+2)^2(K+1)^2(b+1)^2}.
\end{equation*}

With probability at least $99/100$, every learned table is valid.
By \eqref{js:eq:guide-transfer}, it then satisfies
$\mu_{G,t}(v^0)\le w_v\le6\mu_{G,t}(v^0)$ throughout its doubling interval.
Fix any such realization of all tables. Until the final averaging,
all probabilities, expectations, and TV distances are conditional on
this realization; independent prelearning leaves fixed kernels and
fresh annealing randomness.
An effective constant $C_\lambda$ from \cref{js:lem:learned-mixing} ensures error
$\delta$ from every input at most $(n+2)\mu_i$, uniformly for
$0<\lambda_i\le\lambda$.

\subsubsection{Accuracy without independent configurations}
\label{cnt:sec:analysis}

Put $R=\lfloor n/2\rfloor$ and $h_i=\E_{\mu_i}H$.
The empty point mass is 2-warm at $\lambda_0$, since
$Z(\lambda_0)\le e^{m\lambda_0}<2$.
Furthermore $\mu_{i-1}\le e^{R/s}\mu_i\le2\mu_i$.
Comparing the actual input to this warm reference at each block gives
$e_i:=\TV(\operatorname{Law}(M_i),\mu_i)\le(i+1)\delta$.
This is a TV comparison, not a claim that the actual input is warm.

For $i<j$, introduce the observable-biased law
\[
 \rho_i=\frac{H+1}{h_i+1}\mu_i.
\]
It obeys $\rho_i\le(R+1)\mu_i\le(n+2)\mu_{i+1}$.
Both $\rho_i$ and $\mu_i$ therefore mix to error $\delta$ in the
first future block $Q_{i+1}$. Contraction through all remaining
blocks bounds their future expectations of $H_j$ within $2R\delta$.
Write $\E_{\rm ref}$ for expectation under the process started from
$M_i\sim\mu_i$ and evolved by $Q_{i+1},\ldots,Q_j$. Then
\[
 |\E_{\rm ref}[(H_i-h_i)(H_j-h_j)]|
 \le2R(R+1)\delta.
\]
Indeed the left expectation is $(h_i+1)$ times the difference of
these two future expectations; the constant $h_j$ cancels because
the first centered factor has mean zero. The actual pair differs
from this pair by at most $e_i$ in TV. Since the centered product
has magnitude at most $R^2$, this adds at most $2R^2e_i$.
The diagonal term is at most $n/8+R^2e_i$ by
\eqref{cnt:eq:var}. Summing with weights $a_i a_j$ proves
\begin{equation}\label{cnt:eq:mse}
 \E\left(\sum_i a_i(H_i-h_i)\right)^2
 \le\frac n8 A_2+5(R+1)^2(K+1)\delta L^2
 \le\varepsilon^2/2^{18}.
\end{equation}
No conditioning on a specified past matching has been used.

For $F(\beta)=\log Z(e^\beta)$, one has
$F'(\beta)=\E H$ and $F''(\beta)=\Var H\le n/8$.
The right Riemann sum $\sum_i a_i h_i$ thus differs from
$\log Z(\lambda)-\log Z(\lambda_0)$ by at most
$nA_2/16\le\varepsilon^2/2^{20}$.
Also $\log Z(\lambda_0)\le\varepsilon/100$.
Chebyshev's inequality in \eqref{cnt:eq:mse} now gives
\[
 \Prb(|\log\widehat Z-\log Z(\lambda)|>\varepsilon/3)<1/100.
\]
This bound holds for every fixed valid realization of the tables.
Averaging over those realizations and adding the guide-learning
failure probability gives the unconditional accuracy guarantee.

\subsubsection{Work bound}\label{cnt:sec:implementation}

Since $b=O_\lambda(\log(m/\varepsilon))$, all learning stages and table
constructions cost $\widetilde O_\lambda(m+n)$ in total.
The bound \eqref{cnt:eq:starting-activity} controls every small-activity
logarithm in the Glauber burn-ins. Within a doubling interval,
neighbor tables and sums $\sum_{v\sim u}w_v$ remain fixed.
Store these unscaled sums in the free-vertex tree: changing activity
then changes one global multiplier. No edge scan is needed between
two successive fine activity values.

Every block consists of exactly $N=\widetilde O_\lambda(n)$ learned JS steps,
and $K=\widetilde O_\lambda(n/\varepsilon^2)$.
Each step costs $O(\log n)$ work by \cref{js:sec:learned-kernel},
even with invalid but clipped positive guides. Thus, on every run,
\[
 \text{work}\le\widetilde O_\lambda(m+n+KN)
       =\widetilde O_\lambda(n^2/\varepsilon^2).
\]
Here simplicity gives $m\le n(n-1)/2$, so the final bound includes
all preprocessing, guide learning, and initialization.

Since $H_i\le n/2$, repeated squaring evaluates
\eqref{cnt:eq:estimator} in $O(K\log n)$ arithmetic operations.
The guide and statistical failure probabilities sum to less than
$1/50$. Logarithmic error at most $\varepsilon/3$ implies relative
error at most $\varepsilon$, proving \cref{thm:intro-counting}.
The sampling and learning constants are uniform for activities at most
$\lambda_{\max}$; the direct-return branch and
\eqref{cnt:eq:starting-activity} also control small-activity burn-in costs.

\subsection{The learned dynamics and monomer weights}
\label{js:sec:learned-kernel}

\subsubsection{Local moves and the normalization of star updates}
\label{js:sec:normalization}

Fix weights $0<w_v\le1$. Their accuracy affects the running time;
the kernel below preserves $\mu$ for every such table. Its local
proposal weights are
\begin{equation}\label{js:eq:proposal-weights}
\begin{aligned}
 c_w(M,M-uv)&=w_u+w_v &&(uv\in M),\\
 c_w(M,M+uv)&=\lambda(w_u+w_v) &&(u,v\text{ free}),\\
 c_w(M,M-vz+uv)&=\lambda w_v &&(u\text{ free},\ vz\in M).
\end{aligned}
\end{equation}
All other distinct-state weights are zero. Define the reversible
proposal generator by
\begin{equation*}
 (\mathcal G_w f)(M)=\sum_{M'\ne M}
                   c_w(M,M')[f(M')-f(M)].
\end{equation*}
To identify what the weights do, let $H_v$ resample the edges incident
to $v$ conditional on all other edges. After erasing the current edge
at $v$, the number of available neighbors and the conditional
partition function are
\begin{equation*}
 F_v(M)=\sum_{u\sim v}J_u(M)+1-J_v(M),
 \qquad Z_v(M)=1+\lambda F_v(M).
\end{equation*}
The blank choice has weight one and every available edge has weight
$\lambda$. Thus the proposal rates give the exact identity
\begin{equation*}
 \mathcal G_w=\sum_v w_vZ_v(H_v-I).
\end{equation*}
A star at $v$ contributes $w_v$ to a deletion and $\lambda w_v$
to an insertion or exchange. The multiplier $w_vZ_v$ is constant
on each conditional star: it is that star's update rate in
$\mathcal G_w$.

Write $\mu^{v^0}$ for $\mu$ conditioned on $v$ being free, identified
with $\mu_{G-v,\lambda}$. For an outside-star matching, the sum
of the weights of its star completions is $Z_v$. Its law under $\mu$
therefore has density $\mu(v^0)Z_v$ relative to $\mu^{v^0}$. Normalizing gives
\begin{equation}\label{js:eq:clock-normalization}
 \E_{\mu^{v^0}}Z_v=\frac1{\mu(v^0)},
 \qquad \mu(v^0)Z_v=\frac{Z_v}{\E_{\mu^{v^0}}Z_v}.
\end{equation}
Thus $w_v=\mu(v^0)$ normalizes the star rate by its mean under the
free-vertex conditional law. Constant-factor estimates suffice.
The lower bound $w_v\ge \mu(v^0)$ lets a lower-tail bound for $\mu(v^0)Z_v$
control stars with small proposal rates. The upper bound $w_v\le B\mu(v^0)$
controls the total proposal rate and hence the cost of passing to
discrete steps. These are the two uses of the learned monomer
marginals in the WPI proof.

\subsubsection{Discrete steps and their implementation}

The total proposal rate is
\begin{equation}\label{js:eq:rate}
\begin{gathered}
 r_u=\lambda\sum_{v\sim u}w_v,\qquad
 R(M)=\sum_{M'\ne M}c_w(M,M')=X(M)+Y(M),\\
 X(M)=\sum_u r_uJ_u(M),\qquad
 Y(M)=\sum_vw_v(1-J_v(M)).
\end{gathered}
\end{equation}
The two orientations of an insertion and the two endpoints of a deletion
are both included in these sums. For a graph with at least one edge,
$R(M)>0$ for every matching: a nonempty matching has a deletion, and
an empty matching has an insertion. On an edgeless graph, the chain
stays at the empty matching; the following description assumes $m\ge1$.

A step of the \emph{learned Jerrum--Sinclair chain} first holds with
probability $1/2$. Otherwise it proposes $M'$ with probability
$c_w(M,M')/R(M)$ and accepts with probability
$\min\{1,R(M)/R(M')\}$. Equivalently,
\begin{equation}\label{js:eq:learned-kernel}
 \widehat P_w(M,M')=
 \frac{c_w(M,M')}{2\max\{R(M),R(M')\}}\quad(M\ne M'),
\end{equation}
with the diagonal chosen to make each row sum one.
The proposal is generated by selecting a free vertex $u$ with weight
$r_u$ or a matched endpoint $v$ with weight $w_v$. In the first case,
choose a neighbor with weight $w_v$ and propose the corresponding
insertion or exchange; in the second case, propose deletion.
This realizes the weights in \eqref{js:eq:proposal-weights} exactly.

Insertion and deletion have weight ratio $\lambda$, whereas opposite
exchanges about $v$ have equal weight $\lambda w_v$. Therefore
\begin{equation}\label{js:eq:proposal-balance}
 \mu(M)c_w(M,M')=\mu(M')c_w(M',M).
\end{equation}
The symmetric denominator in \eqref{js:eq:learned-kernel} proves detailed
balance. This is the standard Metropolis--Hastings correction for the
normalized proposal~\cite{Hastings70}. The off-diagonal row sum is at
most $1/2$, so the kernel is lazy and positive semidefinite. Additions
and deletions give irreducibility. Each accepted move changes at most
two matching edges and at most two monomer indicators.

\paragraph{Implementation.}
Precompute the $r_u$ and a weighted-neighbor prefix table at every
nonisolated vertex. Maintain mate arrays and two dynamic sum trees,
with weights $r_uJ_u$ and $w_u(1-J_u)$ respectively. Construction takes
$O(m+n)$ arithmetic work. Weighted tree selection and binary search in
a neighbor prefix table each take $O(\log n)$ work.

Put $d_u=r_u-w_u$. The candidate total is computed before accepting:
\begin{equation*}
\begin{array}{c|c}
\text{candidate move}&R(M')-R(M)\\\hline
\text{insert }uv&-d_u-d_v\\
\text{delete }uv&d_u+d_v\\
\text{replace }vz\text{ by }uv&d_z-d_u.
\end{array}
\end{equation*}
These identities follow from $R(M)=\sum_vw_v+\sum_ud_uJ_u(M)$.
On acceptance, update at most two free-status entries and the mate
arrays; on rejection, no entry changes. Thus every discrete step costs
$O(\log n)$ arithmetic operations, for every trajectory and every
positive guide table, including clipped invalid tables. Observations
count both rejections and lazy holds.

\subsubsection{Learning and reusing the weights}
\label{js:sec:learning}

We now show that ordinary Glauber occupation times estimate all the
required marginals simultaneously. We use the weighted monomer covariance
and absolute row-sum bounds of \cref{thm:monomer-input}, stated and
derived from \cite{ConcurrentHolant26} in Appendix~\ref{sec:covariance-proof}.
These bounds also show that a table remains valid, up to a constant
factor, throughout a doubling interval of activities.

\begin{lemma}[Consequences of monomer covariance]\label{cnt:lem:cov}
For $\mu_{G,\lambda}$ with unit monomer weights,
\begin{alignat}{2}
 &\Var H\le n/8,
 &\qquad&\mu(v^0)\ge(1+\lambda\deg_G(v))^{-1},\label{cnt:eq:var}\\
 &\lambda\sum_u\sum_{v\sim u}\mu(u^0)\mu(v^0)\le(1+\lambda/4)n,
 &\qquad&-1\le\frac{d\log \mu(v^0)}{d\log\lambda}\le0.\label{js:eq:guide-transfer}
\end{alignat}
More generally, under any matching law with positive edge activities
and positive monomer weights, let $F=\sum_v c_vJ_v$, where
$0\le c_v\le c_*$ and $c_*>0$. Then
\begin{equation}\label{cnt:eq:monomer-mgf}
 \log\E e^{tF}\le\frac{\E F}{2c_*}(e^{2c_*t}-1)
 \qquad(t\in\mathbb R).
\end{equation}
\end{lemma}
\begin{proof}
Apply the absolute row-sum bound of \cref{thm:monomer-input} to
$H=(n-\sum_vJ_v)/2$ to obtain $\Var H\le n/8$.
Deleting an occupied edge $uv$ gives
$\mu(uv)=\lambda\mu(u^0v^0)$, so
\[
 \lambda\sum_u\sum_{v\sim u}\mu(u^0)\mu(v^0)
 \le2\E H+\lambda\sum_v\mu(v^0)(1-\mu(v^0))\le(1+\lambda/4)n.
\]
Also
$d\mu(v^0)/d\log\lambda=\Cov(J_v,H)
=-\frac12\sum_u\Cov(J_v,J_u)$ lies between
$-\mu(v^0)(1-\mu(v^0))$ and zero. The lower bound on $\mu(v^0)$ follows by
conditioning outside its star.

For the exponential moment, put $g(t)=\log\E e^{tF}$.
Tilting by $e^{tF}$ changes only the positive monomer weights, so
\cref{thm:monomer-input} gives $g''(t)\le2c_*g'(t)$.
Thus $e^{-2c_*t}g'(t)$ is nonincreasing. For $t\ge0$, integrate
$g'(u)\le g'(0)e^{2c_*u}$ over $[0,t]$; for $t<0$, integrate
the reverse inequality over $[t,0]$ and negate.
Both give \eqref{cnt:eq:monomer-mgf}.
\end{proof}

\begin{lemma}[Guides learned by ordinary Glauber dynamics]
\label{js:lem:learning}
For $0<\delta<1/2$, bounded $\widetilde O_\lambda(m+n)$ work suffices
to compute weights $0<w_v\le1$ such that, with probability $1-\delta$,
\begin{equation*}
 \mu(v^0)\le w_v\le3\mu(v^0)\qquad(v\in V).
\end{equation*}
The suppressed factors include $\log(1/\delta)$.
Uniformly for $0<\lambda\le\lambda_{\max}$ with fixed $\lambda_{\max}$,
only additional logarithms of $1/\lambda$ at small activity are needed.
For $0<\eta\le1/2$, fresh Glauber initialization costs another
$\widetilde O_\lambda(m+n)$ work, including $\log(1/\eta)$, and is
$\eta$-accurate conditional on every realized guide table.
\end{lemma}
This initialization guarantee is in TV, not a pointwise warm-start bound.
\begin{proof}
If $m=0$, take $w_v=\mu(v^0)=1$ for every vertex and initialize at the empty
matching. Henceforth assume $m\ge1$.
For this and subsequent proofs in this section, use the Glauber generator
$\cL_{\mathrm{GD}}=m(1+\lambda)(P_{\mathrm{GD}}-\mathrm{Id})$, which deletes
occupied edges at rate one and inserts addable edges at rate $\lambda$.
At the fixed activity $\lambda$, write $\cE$ for the Dirichlet form
$\cE_1$ of Section~\ref{sec:abstract-glauber}, evaluated at
$\mu=\mu_{G,\lambda}$. Thus
\begin{equation*}
\begin{aligned}
 \cE(f,f):=\cE_1(f,f)
 &=\E_\mu\sum_{e\in M}[f(M)-f(M-e)]^2\\
 &=-\langle f,\cL_{\mathrm{GD}} f\rangle_\mu
 =m(1+\lambda)\cE_{P_{\mathrm{GD}}}(f,f).
\end{aligned}
\end{equation*}
By \cref{prop:matching-gap}, its spectral
gap is at least $\Gamma_\lambda:=c_\lambda=e^{-I_\lambda}>0$,
with $I_\lambda$ defined in \eqref{eq:c-lambda}.
We prove that the
asymptotic variance of a monomer indicator is proportional to $\mu(v^0)^2$,
even when $\mu(v^0)$ is small.
For a centered function $g$ and an irreducible reversible generator
$\mathcal A$, write
\[
 \|g\|_{-1,\mathcal A}^2
 :=\langle g,(-\mathcal A)^{-1}g\rangle_\mu
 =\sup_f\{2\langle g,f\rangle_\mu-\cE_{\mathcal A}(f,f)\},
\]
where the inverse is on the mean-zero subspace. The variational
formula also applies to a reducible generator when $g$ is orthogonal
to its kernel.

Let $F$ count free neighbors of $v$ under
$\mu^{v^0}$, and put $Z=1+\lambda F$, $A=\E_{\mu^{v^0}} Z=1/\mu(v^0)$.
The law outside the star of $v$ under $\mu$ has density $\mu(v^0) Z$
relative to $\mu^{v^0}$, and $r=\E[J_v\mid\text{outside the star}]=1/Z$.
By \cref{thm:monomer-input}, $\Var_{\mu^{v^0}} Z\le2\lambda(A-1)$. The identity
\[
 \E_{\mu^{v^0}}\frac1Z=\frac1A+
       \E_{\mu^{v^0}}\frac{(Z-A)^2}{A^2Z}\le\frac{1+2\lambda}{A}
\]
gives $\E_\mu r^2\le(1+2\lambda)\mu(v^0)^2$ and
$\Var_\mu r\le2\lambda \mu(v^0)^2$.
Write $J_v-\mu(v^0)=h+(r-\mu(v^0))$, where $h=J_v-r$.
Let $\cL_{\mathrm{GD},v}$ be the restriction of $\cL_{\mathrm{GD}}$ to updates
of edges incident to $v$, with unchanged rates. Then
$-\cL_{\mathrm{GD},v}h=Zh$, and its Dirichlet form is a sub-sum of that of
$\cL_{\mathrm{GD}}$. The variational formula for the inverse generator gives
\[
 \|h\|_{-1,\cL_{\mathrm{GD}}}^2\le\E[r^2(1-r)],\qquad
 \|r-\mu(v^0)\|_{-1,\cL_{\mathrm{GD}}}^2\le2\lambda \mu(v^0)^2/\Gamma_\lambda.
\]
Indeed $h/Z$ solves the local Poisson equation in each conditional
star. Hence
\begin{equation*}
 \|J_v-\mu(v^0)\|_{-1,\cL_{\mathrm{GD}}}^2\le K_\lambda \mu(v^0)^2,
 \qquad K_\lambda=2(1+2\lambda+2\lambda/\Gamma_\lambda).
\end{equation*}

Put $s_0=m(1+\lambda)$, so
$\cL_{\mathrm{GD}}=s_0(P_{\mathrm{GD}}-\mathrm{Id})$.
To pass to discrete occupation times, use the
standard covariance-series bound for a stationary chain with an
irreducible, reversible, positive semidefinite kernel $P$. For centered $g$,
\[
\begin{aligned}
 \Var\left(N^{-1}\sum_{t=1}^N g(M_t)\right)
 &=\frac{\|g\|_2^2}{N}
   +\frac2{N^2}\sum_{h=1}^{N-1}(N-h)\langle g,P^h g\rangle_\mu\\
 &\le\frac2N\langle g,(I-P)^{-1}g\rangle_\mu.
\end{aligned}
\]
Indeed all spectral coefficients are nonnegative, so the finite sum
is bounded by the corresponding geometric series.
For $P=P_{\mathrm{GD}}$, $(I-P)^{-1}=s_0(-\cL_{\mathrm{GD}})^{-1}$ on centered
functions. Thus
\[
 \Var\left(N^{-1}\sum_{t=1}^N J_v(M_t)\right)
 \le\frac{2s_0K_\lambda \mu(v^0)^2}{N}.
\]
Choose $N=\lceil128s_0K_\lambda\rceil$. Each estimate has relative
error at most $1/2$ except with probability $1/16$.
Use an odd number $R=O(\log(n/\delta))$ of independent blocks, each
preceded by a fresh worst-start Glauber burn-in with error $\delta/(4R)$.
Choose its length using \eqref{eq:discrete-time} and
\eqref{eq:matching-min-mass}.
Coupling their initial states to independent stationary starts and taking
coordinatewise medians puts every median in $[\mu(v^0)/2,3\mu(v^0)/2]$ with
probability at least $1-\delta$. Set $w_v$ to twice the median,
clipped to $[(1+\lambda(n-1))^{-1},1]$.

All $n$ occupation totals can be measured in $O(1)$ work per tick:
store the start of each current free interval, closing it whenever its
vertex changes status. A tick changes at most two statuses. Flush all
counters in $O(n)$ work per block. The burn-ins, measurements, and medians
therefore cost the asserted amount.
Use a fresh burn-in after learning: conditional on every realized guide
table its output is $\eta$-close to $\mu$. An accurate initial law is
not asserted to have a bounded pointwise density ratio.
For the asserted uniformity, $I_\lambda\le I_{\lambda_{\max}}$ and
$K_\lambda\le2(1+2\lambda_{\max}
+2\lambda_{\max}e^{I_{\lambda_{\max}}})$. Thus measurement blocks
have length $O_{\lambda_{\max}}(m)$; only the Glauber burn-ins introduce
the stated small-activity logarithms, as follows from
\eqref{eq:A-lambda} and \eqref{eq:matching-min-mass}.
\end{proof}

\subsection{Warm-start mixing from the star relaxation bound}
\label{js:sec:learned-mixing}
\label{js:sec:proposal-estimates}

The target is the following WPI and its warm-start mixing consequence.

\begin{lemma}[Learned JS mixing]\label{js:lem:learned-mixing}
Fix $\mu(v^0)\le w_v\le B\mu(v^0)$, $w_v\le1$, with constant $B$.
For $0<\rho\le1/2$, the learned kernel satisfies
\begin{equation}\label{js:eq:learned-wpi}
 \Var_\mu f\le
 C_{\lambda,B}n\log^2 n\log^2\frac{n}{\rho}\,
 \cE_{\widehat P_w}(f,f)
 +\rho\operatorname{osc}(f)^2.
\end{equation}
For $W\ge1$, $0<\varepsilon\le1/2$, and $\nu\le W\mu$,
\begin{align}
 \TV(\nu\widehat P_w^N,\mu)&\le\varepsilon,
 \label{js:eq:learned-mixing}\\
 \text{provided}\qquad N&\ge C'_{\lambda,B}n\log^2 n
       \log^2\frac{nW}{\varepsilon}
       \log\frac{eW}{\varepsilon}.
 \label{js:eq:learned-step-bound}
\end{align}
These are actual discrete steps, and the constants are uniform over
$0<\lambda\le\lambda_{\max}$ for every fixed finite $\lambda_{\max}$.
\end{lemma}

Recall that $H_v$ resamples the star at $v$ conditional on all other
edges. Write $\Var_v(f)$ for that conditional variance and
$\mathcal S_G=\sum_v(H_v-I)$.
The proof uses two comparisons, in the order
\[
 \underbrace{\mathcal S_G}_{\text{one update rate per star}}
 \ \longrightarrow\ \mathcal G_w
 \ \longrightarrow\ \widehat P_w.
\]
All auxiliary chains introduced in \cref{sec:star-inequalities} supply
only the bound $T_{\rm star}=O_\lambda(\log^2 n)$ in
\[
 \Var_\mu f\le T_{\rm star}\cE_{\mathcal S_G}(f,f).
\]
We use this bound here and prove it after completing the WPI comparison.

\subsubsection{From star updates to the proposal form}

The two forms differ only by the star rates:
\begin{equation}\label{js:eq:two-star-forms}
 \cE_{\mathcal S_G}(f,f)=\sum_v\E_\mu\Var_v(f),
 \qquad
 \cE_{\mathcal G_w}(f,f)=\sum_v\E_\mu[w_vZ_v\Var_v(f)].
\end{equation}
Because $w_v\ge \mu(v^0)$, it suffices to control the lower tail of
$\mu(v^0)Z_v$.

\begin{lemma}[Small star coefficients]\label{js:lem:clock-tail}
For $L\ge4$ and $V_v=\mu(v^0)Z_v$,
\begin{equation}\label{js:eq:clock-tail}
 \Prb_\mu(V_v<1/L)\le L^{-1}
                    \exp(-(L-1)/(16\lambda)).
\end{equation}
\end{lemma}
\begin{proof}
Under $\mu^{v^0}$, let $F$ count free neighbors of $v$, and put
$Z=1+\lambda F$ and $A=\E_{\mu^{v^0}} Z=1/\mu(v^0)$, as in
\eqref{js:eq:clock-normalization}. Applying \eqref{cnt:eq:monomer-mgf}
with $c_*=1$ and $t=-(\log2)/2$, Chernoff's inequality gives
\[
 \Prb_{\mu^{v^0}}(F\le\E_{\mu^{v^0}} F/2)
 \le e^{-(1-\log2)\E_{\mu^{v^0}} F/4}\le e^{-\E_{\mu^{v^0}} F/16}.
\]
The event $\mu(v^0)Z<1/L$ is empty unless $A>L$; otherwise it implies
$F<\E_{\mu^{v^0}} F/2$ and $\E_{\mu^{v^0}} F>(L-1)/\lambda$.
The outside-star law under $\mu$ has density $Z/A<1/L$ on this
event, proving \eqref{js:eq:clock-tail}.
\end{proof}

\begin{lemma}[WPI for the proposal form]\label{js:lem:proposal-wpi}
If $w_v\ge \mu(v^0)$, then for $0<s\le1/2$,
\begin{equation*}
 \Var_\mu f\le a_\lambda(s)\cE_{\mathcal G_w}(f,f)
                     +s\operatorname{osc}(f)^2,
 \qquad
 a_\lambda(s)=O_\lambda\!\left(
       \log^2 n\log\frac{n}{s}\right).
\end{equation*}
The constants can be chosen uniformly for $0<\lambda\le\lambda_{\max}$,
for every fixed finite $\lambda_{\max}$.
\end{lemma}

\begin{proof}
By \eqref{js:eq:two-star-forms}, on $\{w_vZ_v\ge1/L\}$ the star
variance is at most $L$ times its proposal contribution. On the
complement, use $\Var_v(f)\le\operatorname{osc}(f)^2/4$.
Since $w_v\ge\mu(v^0)$, \cref{js:lem:clock-tail} bounds the probability
of this complement. Summing over $v$ and applying the star bound
with $T_{\rm star}=d_\lambda^{-1}\log^2n$ from \cref{js:lem:stars} gives
\[
 \Var f\le T_{\rm star}L\cE_{\mathcal G_w}(f,f)
 +\frac{nT_{\rm star}}{4L}e^{-(L-1)/(16\lambda)}
                         \operatorname{osc}(f)^2.
\]
Choosing $L=O_\lambda(\log(n/s))$ sufficiently large makes the second
coefficient at most $s$. Uniformity follows by using
$d_{\lambda_{\max}}$ and $\lambda_{\max}$ in these estimates.
\end{proof}

\subsubsection{From proposal rates to discrete mixing}

The second comparison controls the state-dependent denominator in
\eqref{js:eq:learned-kernel}.

For a threshold $b$, call an undirected transition $\{M,M'\}$ with
$c_w(M,M')>0$ \emph{good} if $R(M)\le b$ and $R(M')\le b$, and
\emph{bad} otherwise. On every good transition, \eqref{js:eq:learned-kernel}
gives $c_w(M,M')\le 2b\,\widehat P_w(M,M')$.
By reversibility \eqref{js:eq:proposal-balance}, the total exit capacity
of states with $R>b$ counts each transition crossing this threshold
once and each transition between two such states twice. It therefore
bounds the bad-transition contribution after replacing each squared
increment by $\operatorname{osc}(f)^2$. Consequently,
\begin{equation}\label{js:eq:proposal-discrete-comparison}
 \cE_{\mathcal G_w}(f,f)
 \le2b\cE_{\widehat P_w}(f,f)
   +\E_\mu[R\one_{\{R>b\}}]\operatorname{osc}(f)^2.
\end{equation}
The required upper tail uses the other half of the guide condition,
$w_v\le B\mu(v^0)$.

\begin{lemma}[Total proposal weight]\label{js:lem:total-weight}
Suppose $\mu(v^0)\le w_v\le B\mu(v^0)$ for a constant $B\ge1$. Then
\begin{equation}\label{js:eq:mean-work}
 \E_\mu R\le B(5+\lambda)n/4,
\end{equation}
and, for $C_n=B(1+2\lambda)n$,
\begin{equation}\label{js:eq:exp-work}
 \E_\mu e^{R/C_n}\le e^e.
\end{equation}
In particular, for every $b\ge0$,
\begin{equation}\label{js:eq:total-weight-tail}
 \E_\mu[R\one_{\{R>b\}}]\le e^e(b+C_n)e^{-b/C_n}.
\end{equation}
\end{lemma}
\begin{proof}
Write $R=X+Y$ as in \eqref{js:eq:rate}. By
\cref{cnt:lem:cov},
\[
 \E X\le B(1+\lambda/4)n,\qquad
 \E Y\le B\sum_v\mu(v^0)(1-\mu(v^0))\le Bn/4,
\]
which proves \eqref{js:eq:mean-work}. For the exponential moment,
use $X=\sum_u r_uJ_u$ with $r_u=\lambda\sum_{v\sim u}w_v\le\lambda Bn$.
Apply \eqref{cnt:eq:monomer-mgf} with $c_*=\lambda Bn$ and $t=1/C_n$.
Writing $x=2\lambda Bn/C_n\le1$ gives
\[
 \log\E e^{X/C_n}
 \le\frac{\E X}{C_n}\frac{e^x-1}{x}\le e-1,
\]
since $\E X/C_n\le1$ and $(e^x-1)/x\le e-1$ for $0<x\le1$.
Together with $Y\le Bn\le C_n$, this proves \eqref{js:eq:exp-work}.
Finally integrate the
exponential upper tail of $R$:
\[
 \E[R\one_{\{R>b\}}]
 =b\Prb(R>b)+\int_b^\infty\Prb(R>t)\,dt
 \le e^e(b+C_n)e^{-b/C_n}.
\]
\end{proof}

\begin{proof}[Proof of \cref{js:lem:learned-mixing}]
Combining \eqref{js:eq:proposal-discrete-comparison} with
\cref{js:lem:proposal-wpi} yields
\begin{equation}\label{js:eq:learned-comparison}
 \Var f\le2b a_\lambda(s)\cE_{\widehat P_w}(f,f)
 +\left(s+a_\lambda(s)\E_\mu[R\one_{\{R>b\}}]\right)
                      \operatorname{osc}(f)^2.
\end{equation}
Set $s=\rho/2$, $C_n=B(1+2\lambda)n$, and choose
\[
 U=\max\left\{1,\frac{4e^eC_na_\lambda(\rho/2)}\rho\right\},
 \qquad b=2C_n(1+\log U).
\]
By \eqref{js:eq:total-weight-tail} and
$e^{-2}(3+2\log U)\le2U$ for $U\ge1$,
\[
 a_\lambda(\rho/2)\E[R\one_{\{R>b\}}]\le\rho/2.
\]
Moreover $b=O_{\lambda,B}(n\log(n/\rho))$.
Substitution in \eqref{js:eq:learned-comparison} proves
\eqref{js:eq:learned-wpi}. The threshold $b$ is used only in the proof.

\emph{Discrete variance contraction.}
Let $h_t=d(\nu\widehat P_w^t)/d\mu$ and $V_t=\Var_\mu h_t$.
Reversibility gives $h_{t+1}=\widehat P_wh_t$, and Markov contraction
preserves $0\le h_t\le W$. In particular $V_0\le W-1$ and
$\operatorname{osc}(h_t)\le W$. Since the kernel is positive
semidefinite,
\[
 V_t-V_{t+1}
 =\langle h_t,(I-\widehat P_w^2)h_t\rangle_\mu
 \ge\cE_{\widehat P_w}(h_t,h_t).
\]
Take $\rho=\varepsilon^2/W^2\le1/4$ in
\eqref{js:eq:learned-wpi}, and enlarge its coefficient $D$
to be at least one. Then
\[
 V_{t+1}\le(1-D^{-1})V_t+\varepsilon^2/D,
 \qquad V_t\le e^{-t/D}(W-1)+\varepsilon^2.
\]
After $O(D\log(eW/\varepsilon))$ steps, $V_t\le2\varepsilon^2$.
The inequality $\TV(\nu\widehat P_w^t,\mu)\le\sqrt{V_t}/2$
proves \eqref{js:eq:learned-mixing}--\eqref{js:eq:learned-step-bound}.
The proposal WPI is uniform on each bounded activity interval, and
$C_n\le B(1+2\lambda_{\max})n$ there. Consequently both constants in
the lemma are uniform on that interval; no positive lower bound on
the activity is needed.
\end{proof}

\subsection{A comparison proof of the star relaxation bound}
\label{sec:star-inequalities}

We prove the star relaxation bound used in
\cref{js:sec:learned-mixing}. Write
$T_{\rm star}^G(\lambda)=\gap(-\mathcal S_G)^{-1}$, with update rate
one per vertex. On a singleton state space the variance bounds are
trivial. The auxiliary kernels below are used only in the proof.

\begin{lemma}[Star spectral gap]\label{js:lem:stars}
For every simple graph at activity $\lambda>0$,
\begin{equation}\label{js:eq:star-gap}
 \gap(-\mathcal S_G)\ge d_\lambda\log^{-2}n,\qquad
 d_\lambda=
 \frac{c_0e^{-I_{2\lambda}}}
 {\max\{1,\lambda\}(1+\lambda e^{I_\lambda})},
\end{equation}
where $c_0>0$ is absolute. In particular $d_\lambda\ge d_{\lambda_{\max}}>0$
whenever $0<\lambda\le\lambda_{\max}$.
\end{lemma}

There are two branches in the proof. The first uses the gap of $\cL_{\mathrm{GD}}$
to control a matching-to-cut-to-matching kernel, called
Bip Coding. The second controls bipartite star updates through a
small-activity down--up walk and field dynamics. They combine as follows:
\[
\begin{array}{ccccc}
 \text{Glauber gap}&\longrightarrow&\text{Bip Coding}&&\\
 &&\searrow&&\\[-1mm]
 &&&\text{general-graph star}&\\[-1mm]
 &&\nearrow&&\\
 \text{down--up / one-shore star}
 &\xrightarrow{\ \text{field}\ }&\text{bipartite star}.&&
\end{array}
\]
We establish these comparisons first, postponing the Bip Coding and
small-activity down--up gap proofs to \cref{js:sec:auxiliary-proofs}.

\subsubsection{Bip Coding and the reduction to bipartite stars}

Given $M$, independently orient each occupied edge across a random
cut, and assign each free vertex to either side with equal probability.
If $L$ is one side of the cut, the joint law is
\begin{equation}\label{js:eq:bip-coding-law}
 \varpi_\lambda(M,L)=
 \frac{2^{-n}(2\lambda)^{|M|}}{Z_G(\lambda)}
          \one\{M\subseteq E(L,L^c)\}.
\end{equation}
Its matching marginal is $\mu_{G,\lambda}$, and
$M\mid L\sim\mu_{G[L,L^c],2\lambda}$. Define $P_{\rm Bip}$ by
sampling $L\mid M$ and then a new matching conditional on $L$.
Its Dirichlet form and relaxation time satisfy
\begin{equation}\label{js:eq:bip-coding-gap}
 \cE_{P_{\rm Bip}}(f,f)=\E_L\Var(f(M)\mid L),
 \qquad
 T_{\rm Bip}:=\gap(P_{\rm Bip})^{-1}
       \le1+\lambda e^{I_\lambda}.
\end{equation}
The identity follows from conditional resampling; the gap bound is
proved in \cref{js:sec:auxiliary-proofs} using the already established
gap lower bound $\Gamma_\lambda=e^{-I_\lambda}$ for $\cL_{\mathrm{GD}}$.

For a fixed cut, apply the star Poincar\'e inequality on its bipartite
graph at activity $2\lambda$. Let
$\mathcal O_v(M)=\{e\in M:v\notin e\}$ be the matching configuration
outside the star at $v$. The conditional variance comparison
needed to average over cuts is explicit:
\[
 \E_{\varpi_\lambda}\Var(f(M)\mid\mathcal O_v,L)
 \le\E_\mu\Var(f(M)\mid\mathcal O_v).
\]
This is the law of total variance, and summing over $v$ gives
\begin{equation*}
 \E_L\cE_{\mathcal S_{G[L,L^c]},\,2\lambda}(f,f)
 \le\cE_{\mathcal S_G,\,\lambda}(f,f).
\end{equation*}
Combining these identities yields the comparison of relaxation times:
\begin{equation}\label{js:eq:general-bip-comparison}
 T_{\rm star}^G(\lambda)
 \le T_{\rm Bip}\sup_L T_{\rm star}^{G[L,L^c]}(2\lambda).
\end{equation}
Cuts with a singleton matching space contribute zero variance.
Thus only the bipartite star bound remains to be supplied.

\subsubsection{Down--up is a one-shore star update}

On a bipartite graph $A\sqcup B$ with $k=|B|>0$, encode a matching as
\[
 Y(M)=M\cup\{b^0:b\in B\text{ is free in }M\}.
\]
Each $b\in B$ contributes exactly one symbol: either its matched
edge or its blank $b^0$. The resulting law is supported on $k$-element
sets. A down--up step deletes one uniformly chosen symbol and
resamples it conditional on the remaining $k-1$ symbols. The missing
symbol identifies the vertex $b$ whose star is resampled, so exactly
\begin{equation}\label{js:eq:down-up-star}
 P_{\rm DU}=\frac1k\sum_{b\in B}H_b,
 \qquad \mathcal S_B:=\sum_{b\in B}(H_b-I)=k(P_{\rm DU}-I).
\end{equation}
In particular a $1/(2k)$ down--up gap means a $1/2$ gap for the
one-shore star generator. The technical input, proved in
\cref{js:sec:auxiliary-proofs}, is
\begin{equation}\label{js:eq:small-activity-input}
 \lambda_*=[112(1+\log(n+1))]^{-2},\qquad
 \cE_{\mathcal S_B}(f,f)\ge\tfrac12\Var f
 \quad(0<\alpha\le\lambda_*).
\end{equation}
It holds on every bipartite graph with at most $n$ vertices, at
activity $\alpha$, including all residual graphs. Since
$\cE_{\mathcal S_{A\sqcup B}}\ge\cE_{\mathcal S_B}$, the corresponding
all-star relaxation time is at most two as well.

\subsubsection{Field dynamics transfers the small-activity bound}

Fix a bipartite graph at activity $\alpha$ and $0<t\le\alpha$.
Given $M$, freeze each occupied edge independently with probability
$1-t/\alpha$, obtaining $F_t$, and resample the remaining matching
conditional on $F_t$. Denote this field kernel by
$P_{\rm fld}^{\alpha,t}$. Its conditional unfrozen law is the matching
law at activity $t$ on the graph obtained by deleting the endpoints
of $F_t$. Thus
\[
 \cE_{P_{\rm fld}^{\alpha,t}}(f,f)
 =\E\Var(f(M)\mid F_t).
\]
The field-dynamics estimate \cite[Theorem 1.16]{CCCYZ25}, with
retained fraction $1-t/\alpha$ and rate
$C(s)=\widehat C(\alpha(1-s))$ from \cref{prop:holant-edge}, gives
\begin{equation*}
 \frac t\alpha e^{-(I_\alpha-I_t)}\Var f
 \le\cE_{P_{\rm fld}^{\alpha,t}}(f,f).
\end{equation*}
Here $I_\alpha-I_t=\int_t^\alpha(\widehat C(u)-1)\,du/u$ by
\eqref{eq:c-lambda}; substituting $u=\alpha(1-s)$ in the cited
integral gives the displayed factor.
At $t=\alpha$ the kernel resamples the whole law, so equality holds.
For $t_0=\min\{\alpha,\lambda_*\}$, the residual star bound
\eqref{js:eq:small-activity-input} gives the other half of the comparison:
\begin{equation*}
 \frac{t_0}{\alpha}e^{-(I_\alpha-I_{t_0})}\Var f
 \le\E\Var(f(M)\mid F_{t_0})
 \le2\cE_{\mathcal S_{A\sqcup B},\,\alpha}(f,f).
\end{equation*}
For the last inequality, apply the star inequality conditionally on
$F_{t_0}$, then use
$\E\Var(f\mid\mathcal O_v,F_{t_0})
\le\E\Var(f\mid\mathcal O_v)$ at each vertex. A vertex incident to
a frozen edge has zero residual star variance. Therefore
\begin{equation}\label{js:eq:bip-star-relaxation}
 T_{\rm star}^{A\sqcup B}(\alpha)
 \le\frac{2\alpha}{t_0}e^{I_\alpha-I_{t_0}}
 =O_\alpha(\log^2 n).
\end{equation}
\begin{proof}[Proof of \cref{js:lem:stars} from the comparison bounds]
Apply \eqref{js:eq:bip-star-relaxation} at $\alpha=2\lambda$ in
\eqref{js:eq:general-bip-comparison}, and use
\eqref{js:eq:bip-coding-gap}. With
$t_1=\min\{2\lambda,\lambda_*\}$, this gives
\[
 \gap(-\mathcal S_G)\ge
 \frac{t_1}{4\lambda}
 \frac{e^{-(I_{2\lambda}-I_{t_1})}}
      {1+\lambda e^{I_\lambda}}.
\]
Since $t_1/(4\lambda)\ge\lambda_*/(4\max\{1,\lambda\})$ and
\[
 \lambda_*\ge\frac1{\log^2 n}
 \left[\frac{\log2}{112(1+\log3)}\right]^2,
\]
we obtain \eqref{js:eq:star-gap} with
$c_0=\tfrac14[\log2/(112(1+\log3))]^2$.
\end{proof}

\subsubsection{Auxiliary spectral-gap estimates}
\label{js:sec:auxiliary-proofs}

We complete the argument by proving the two auxiliary gap estimates.

\paragraph{Glauber controls Bip Coding.}

We prove the gap bound in \eqref{js:eq:bip-coding-gap}. Under the joint
law \eqref{js:eq:bip-coding-law}, let
$Cg(M)=\E[g(L)\mid M]$. Conditional on $M$, the free vertex signs
are independent fair signs. For each addable edge $e$, let $\chi_e$
be the product of its endpoint signs. These products are orthonormal
for distinct addable edges, and
\[
 Cg(M+e)-Cg(M)=-\E[g\chi_e\mid M].
\]
Indeed adding $e$ conditions its two signs to be opposite.
Bessel's inequality gives
\[
 \sum_{e\text{ addable}}[Cg(M+e)-Cg(M)]^2
 \le\Var(g(L)\mid M).
\]
Recall that $\cE$ is the unit-deletion Glauber form; equivalently,
it sums insertion increments with rate $\lambda$.
The gap lower bound $\Gamma_\lambda=e^{-I_\lambda}$
therefore gives
\[
 \Gamma_\lambda\Var(Cg)
 \le\cE(Cg,Cg)
 \le\lambda\E\Var(g\mid M)
 =\lambda(\Var g-\Var(Cg)).
\]
On centered functions, this bounds the squared norm of $C$ by
$\lambda/(\Gamma_\lambda+\lambda)$. The adjoint conditional expectation
has the same norm and $P_{\rm Bip}=CC^*$, so
\[
 \gap(P_{\rm Bip})\ge
 \frac{\Gamma_\lambda}{\Gamma_\lambda+\lambda},
 \qquad T_{\rm Bip}\le1+\lambda e^{I_\lambda},
\]
as asserted.

\paragraph{The small-activity down--up gap.}

The encoded walk includes monomer blanks as well as occupied edges.
Its gap therefore requires a joint covariance bound; edge covariance
alone does not suffice. We use the following consequence of the
all-activity bound proved in Appendix~\ref{sec:covariance-proof}.

\begin{proposition}[Small-activity joint covariance]
\label{js:prop:small-activity-covariance}
For every finite simple graph and $0<\lambda\le1/16$, the joint vector
$Z=(X,J)$ of dimer and monomer indicators under $\mu_{G,\lambda}$ satisfies
\[
 \Cov(Z)\preceq(1+28\sqrt\lambda)\Diag(\E Z).
\]
The same bound holds after every feasible edge pinning, with conditional
marginal means and deterministic coordinates omitted.
\end{proposition}
\begin{proof}
By \cref{thm:joint-covariance}, it suffices to bound its coefficient.
For $0<\lambda\le1/16$,
\[
 \eta(\lambda)-1
 =39\lambda+12\sqrt{\lambda(1+3\lambda)}
 \le12\sqrt\lambda+60\lambda
 \le27\sqrt\lambda<28\sqrt\lambda.
\]
\end{proof}

Return to the encoding in \eqref{js:eq:down-up-star}, with
$k=|B|$ and indicator vector $Y$, at activity
$0<\lambda'\le\lambda_*<1/16$. Put $\delta=28\sqrt{\lambda'}$.
Conditioning on the presence of
encoded edges or blanks leaves a matching problem on an induced
bipartite graph. The remaining encoded coordinates form a subset of
its joint dimer and monomer coordinates. Thus
\cref{js:prop:small-activity-covariance} gives, in every such conditional law,
\[
 \Cov(Y)\preceq(1+\delta)\Diag(\E Y).
\]

Suppose first that $k\ge2$ and the encoded law is not a point mass.
For a conditional law on $\ell$-element sets, $\ell\ge2$, put
$p_i=\E Y_i$, $p_{ih}=\E Y_iY_h$, omitting zero-marginal coordinates.
Its pair walk has
$K_{ih}=p_{ih}/((\ell-1)p_i)$ for $i\ne h$ and $K_{ii}=0$.
It is reversible for $p/\ell$, and
\[
 \Diag(p)^{-1}\Cov(Y)=I+(\ell-1)K-\one p^\top.
\]
On functions orthogonal to constants under $p/\ell$, the covariance
bound therefore gives
$\lambda_2(K)\le\delta/(\ell-1)$.
The standard local-to-global gap product of
\cite[Theorem 3.1]{AL20}, applied to the down--up kernel $P_{\rm DU}$,
now yields
\[
 \gap(P_{\rm DU})\ge\frac1k
 \prod_{\ell=2}^k
 \left(1-\frac{\delta}{\ell-1}\right).
\]
By \eqref{js:eq:down-up-star}, $\mathcal S_B=k(P_{\rm DU}-I)$.
Each product deficit is at most
$[4(1+\log(n+1))(\ell-1)]^{-1}$, and their sum is at most $1/4$.
Consequently
\begin{equation*}
 \cE_{\mathcal S_B}(f,f)\ge\tfrac12\Var f.
\end{equation*}
For $k=1$, the down--up kernel resamples the entire matching, so the
same bound holds directly; for $k=0$, or any singleton state space,
all variances vanish. The roles of the two shores may be exchanged.

\section{Discussion and open problems}\label{sec:discussion}

The seminal work of Jerrum and Sinclair~\cite{JerrumSinclair}
established polynomial-time tractability of approximate sampling and
counting matchings, and helped shape the modern theory of Markov-chain
algorithms. Our work strengthens this tractability picture: at every
fixed activity $\lambda>0$, ordinary single-edge Glauber dynamics
achieves near-linear mixing and sampling time on every simple graph,
without any degree restriction. This is particularly noteworthy in
the large-degree regime, where deteriorating correlation decay gives
the monomer--dimer model a near-critical character~\cite{BGKNT07}.
Despite this loss of degree-uniform correlation decay, a natural
local dynamics still samples in time nearly proportional to the
input size.

Together with our work-efficient sublinear-depth sampler and faster
approximate counting algorithm, this result opens several directions
for further research. We discuss six below; unless stated otherwise,
the activity $\lambda>0$ is fixed.

\paragraph{A unified MLSI criterion.}
Localization provides a common framework for mixing
analysis~\cite{CE22}, yet existing entropy criteria use different
pinning, external-field, and marginal assumptions, sometimes for
different update chains~\cite{CLV21,CFYZ22,AJKPV22}.
Can a quantitative MLSI criterion encompass these regimes and the
scalar-field hypotheses of \cref{thm:intro-boolean}?
For ordinary matching GD, a natural benchmark is an inverse MLSI
constant $O_\lambda(m)$, which would give
$O_\lambda(m[\log n+\log(1/\varepsilon)])$ mixing.
The broader question is which assumptions allow MLSI to avoid the
loss from small occupied marginals. Such an improvement must use a
sharper entropy argument: the marginal logarithm in the general LSI
bound is necessary even for a Bernoulli law.

\paragraph{Activity dependence and the perfect-matching limit.}
Can our Glauber mixing and sampling bounds retain near-linear
dependence on graph size with only polynomial dependence on $\lambda$?
A more ambitious goal is a monomer--dimer sampler whose running time,
for $\lambda\ge1$, is polynomial in $n,m,\log\lambda$, and
$1/\varepsilon$.
Such a sampler would yield approximate uniform sampling of perfect
matchings on every graph admitting one: conditional on being perfect,
a monomer--dimer matching is uniform, and
$\lambda\ge 2^m/\varepsilon$ already puts probability at least
$1-\varepsilon$ on perfect matchings.
Polynomial dependence on $\lambda$ alone would not suffice for this
reduction, since exponentially large activity can be necessary when
near-perfect matchings vastly outnumber perfect
matchings~\cite[Section~5.4]{JerrumBook}.
The stronger goal cannot follow merely from a sharper worst-start
mixing bound for ordinary single-edge GD: even on a four-cycle,
rare deletions force $\tmix(1/4)=\Omega(\lambda)$ as
$\lambda\to\infty$.
It therefore calls for new sampling ideas beyond improving that
mixing bound.

\paragraph{Sampling matchings in RNC.}
Can the monomer--dimer law on every simple graph be sampled with
polylogarithmic depth and polynomial work at inverse-polynomial
total-variation error, ideally with near-linear work?
This complements RNC maximum matching \cite{MVV87} and the
perfect-matching/permanent questions \cite{Teng95,LY25}; Teng's barrier
concerns Broder--Jerrum--Sinclair simulation, not sampling itself.
The open problem concerns sampling a matching; approximate counting
may be easier to parallelize.

Unlike sequential equivalence \cite{JVV86}, the parallel reductions
are asymmetric. Nonadaptive annealing \cite{HKLYZ26}
counts using one batch of samples at prescribed activities; RNC sampling
along the schedule therefore gives RNC counting. The reverse reductions
\cite{AGR24,ABCHKLV26} provide sublinear, but not polylogarithmic,
expected-round guarantees, using conditional marginals under arbitrary
pinning. Even if sufficiently accurate pinned counts supply these
marginals in RNC, those reductions do not give RNC sampling.

\paragraph{Exact and local sampling of matchings.}
Perfect-sampling approaches based on spatial mixing~\cite{FGY22,AJ22},
coupling from the past (CFTP)~\cite{PW96}, and coupling toward the past
(CTTP)~\cite{FGWWY25} motivate the following questions for matchings
on arbitrary simple graphs, without a degree restriction.
Can one generate $M\sim\mu_{G,\lambda}$ exactly in expected
$\wtO_\lambda(m+n)$ time?
For local sampling, locality should concern query work rather than
graph radius: even a radius-one neighborhood can contain the entire graph.
A concrete goal is to answer $k$ edge-occupation queries in expected
$\wtO_\lambda(k)$ total online work, uniformly in the maximum degree,
where $\wtO$ suppresses polylogarithmic factors in~$n$.
For any fixed feasible pinning $\tau$, the answers must have the
joint law obtained by revealing a single exact sample from
$\mu_{G,\lambda}^{\tau}$.
These questions are particularly compelling in the unbounded-degree
regime, including dense graphs, where exponential spatial-mixing
bounds deteriorate with the degree~\cite{BGKNT07}.

\paragraph{Counting below the annealing baseline.}
Recent subquadratic counting results for spin
systems~\cite{AFFGW25,CCLZ26} suggest a corresponding question for
matchings. Here the annealing baseline is
$\wtO_\lambda(mn/\varepsilon^2)$ after input
preprocessing~\cite{SVV09,HKLYZ26}.
When $m=\Theta(n^2)$, our $\wtO_\lambda(n^2/\varepsilon^2)$ bound
saves a factor $n$ in the polynomial part of this baseline and is
near-linear in the input size at constant accuracy.
Can a polynomial saving over $mn$ be obtained across all graph
densities, apart from the input cost? In particular, for every graph
with $m=\Theta(n)$, can counting be achieved in
$\wtO_\lambda(n^{2-\delta})$ work at constant accuracy, for some
$\delta>0$?
The variance-reduced and aggregate estimators of~\cite{CCLZ26} suggest
a route through perfect marginal sampling, linking this question to
the preceding one. A matching analogue would need sufficiently cheap
conditional samples and, for aggregation, suitable stopping-time
control; the existing spin-system results do not automatically supply
these guarantees on degree-unbounded graphs.

\paragraph{Deterministic approximate counting.}
Bayati, Gamarnik, Katz, Nair, and Tetali~\cite{BGKNT07} give a
deterministic FPTAS for bounded-degree graphs and a
subexponential-time approximation scheme on general graphs.
Can $Z_G(\lambda)$ be approximated deterministically within relative
error $\varepsilon$ in time polynomial in $n$ and $1/\varepsilon$,
without a degree restriction?
Near-linear Glauber mixing makes derandomization particularly natural,
but does not by itself supply an FPTAS.
CTTP provides a concrete route from efficient sampling to
derandomization~\cite{FGWWY25}: under suitable conditions, truncating a
local marginal sampler leaves only logarithmically many random choices
from constant-size domains, which can be enumerated for deterministic
counting. Can one obtain such a sampler for matchings under feasible
pinnings, with controlled truncation error and a polynomial enumeration
cost whose exponent is independent of the maximum degree?
The challenge is to bound this local exploration uniformly in degree,
despite the deterioration of correlation decay~\cite{BGKNT07}.

\paragraph{Disclosure of AI use.}
The main ideas of this work were developed with assistance from OpenAI's
Codex. The human authors contributed original ideas in developing the
AI-generated initial drafts and through continued interaction with the
AI. They worked through the proofs in detail and refined, restructured,
and rewrote the arguments.
AI assistance was also used in preparing, auditing, and revising the
manuscript. The authors take full responsibility for its mathematical
content and final presentation.

\begingroup
\small
\setlength{\emergencystretch}{2em}
\raggedright
\urlstyle{same}
\AddToHookNext{env/thebibliography/begin}{%
  \phantomsection
  \addcontentsline{toc}{section}{References}%
  \setlength{\itemsep}{2pt}%
}
\bibliographystyle{plainurl}
\bibliography{references}
\endgroup

\appendix

\section{Joint covariance of monomer and dimer indicators}
\label{sec:covariance-proof}

We prove an all-activity joint covariance bound for monomer and dimer
indicators. Its small-activity consequence,
\cref{js:prop:small-activity-covariance}, controls the encoded down--up
walk in the counting proof. The weighted monomer bound
(\cref{thm:monomer-input}) also underlies the learning and concentration
estimates in Section~\ref{js:sec:learning}.
We use the harmonic surrogate of
\cite[Definition 17]{ConcurrentHolant26} and the monomer covariance
estimates in \cref{thm:monomer-input}, then adapt their random-revealing
argument to a corrected monomer statistic.

For $M\sim\mu=\mu_{G,\lambda}$, let $X_e=\one_{\{e\in M\}}$ and
$J_u=1-\sum_{e\sim u}X_e$. Write $u^0=\{J_u=1\}$,
$u^1=\{J_u=0\}$, and $\mu(e)=\E_\mu X_e$.
Concatenated events denote intersections, and event superscripts
denote conditioning, so
$\mu^{u^0}(v^0)=\Prb_\mu(J_v=1\mid J_u=1)$.
Deleting an occupied edge $uv$ gives
$\mu(uv)=\lambda\mu(u^0v^0)$.

\begin{theorem}[Covariance bound for monomer and dimer indicators]
\label{thm:edge-cov}\label{thm:joint-covariance}
For every finite simple graph $H$ and activity $\lambda>0$, let
$\mu=\mu_{H,\lambda}$ and let $Z=(X,J)$ be the joint vector of its
dimer and monomer indicators. Its covariance matrix,
including mixed covariances, satisfies
\begin{equation*}
 \Cov(Z)\preceq\eta(\lambda)\Diag(\E Z),
 \qquad
 \eta(\lambda)=
 \left(\sqrt{1+3\lambda}+6\sqrt\lambda\right)^2,
\end{equation*}
where
$\Diag(\E Z)=\Diag((\mu(e))_{e\in E(H)},
(\mu(v^0))_{v\in V(H)})$ uses marginal means. In particular,
\begin{equation}\label{eq:covariance}
 \Var_\mu\left(\sum_e a_eX_e\right)
 \le\eta(\lambda)\sum_e\mu(e)a_e^2
 \qquad(a\in\mathbb R^{E(H)}).
\end{equation}
Both assertions hold after every feasible edge pinning, with the
conditional marginal means and the remaining random coordinates.
Moreover, $\eta(\lambda)=1+12\sqrt\lambda+O(\lambda)$ as
$\lambda\downarrow0$.
\end{theorem}

Fix a finite simple graph $G=(V,E)$ and $\mu=\mu_{G,\lambda}$,
$\lambda>0$, using the notation from the section opening.
Write $Z=(X,J)$ and define
\[
  D_E:=\Diag\bigl((\mu(e))_{e\in E}\bigr),\qquad
  D_V:=\Diag\bigl((\mu(u^0))_{u\in V}\bigr),\qquad
  \Diag(\E Z)=\Diag(D_E,D_V).
\]
\paragraph{Reducing joint covariance to two variance bounds.}

The monomer indicators are affine functions of the dimer indicators.
The incidence correction below uses this relation to convert the bound
for $(X,J)$ into a bound for $X$ with a modified quadratic form.
We decompose the corrected statistic into a centered dimer term
$\mathrm I$ and a monomer correction $\mathrm{II}$, then bound their
variances separately.

\begin{lemma}[Reduction to two variance estimates]\label{cov-lem:reduction}
Let $I_G\in\{0,1\}^{V\times E}$ be the unsigned incidence matrix and set
\[
  A:=D_E^{-1}+I_G^{\mathsf T}D_V^{-1}I_G.
\]
For fixed $\theta\in\mathbb R^E$, define
\begin{align*}
  \mathrm I(X)
  &:=\sum_{uv\in E}\theta_{uv}Y_{uv},
  &Y_{uv}&:=X_{uv}-\lambda J_uJ_v,
  \\
  \mathrm{II}(X)
  &:=\lambda\sum_{uv\in E}\theta_{uv}J_uJ_v
  -\sum_{u\in V}
    \frac{\sum_{e\sim u}\mu(e)\theta_e}{\mu(u^0)}J_u.
\end{align*}
Suppose that, for every $\theta\in\mathbb R^E$,
\begin{align}
  \Var_\mu[\mathrm I]
  &\le\eta_1(\lambda)\left[
    \sum_{e\in E}\mu(e)\theta_e^2
    +\sum_{u\in V}
      \frac{\bigl(\sum_{e\sim u}\mu(e)\theta_e\bigr)^2}{\mu(u^0)}
  \right],\label{cov-eq:first-target}\\
  \Var_\mu[\mathrm{II}]
  &\le\eta_2(\lambda)\sum_{e\in E}\mu(e)\theta_e^2.
  \label{cov-eq:second-target}
\end{align}
Then
\[
  \Cov(Z)
  \preceq
  \bigl(\sqrt{\eta_1(\lambda)}+\sqrt{\eta_2(\lambda)}\bigr)^2
  \Diag(\E Z).
\]
\end{lemma}

\begin{proof}
The edgeless case is immediate, so assume $E\ne\varnothing$.
All entries of $D_E$ and $D_V$ are positive. Since $J=\mathbf1-I_GX$,
\[
  \Cov(Z)
  =\begin{pmatrix}I\\-I_G\end{pmatrix}
    \Cov(X)
    \begin{pmatrix}I&-I_G^{\mathsf T}\end{pmatrix}.
\]
The weighted Schur complement identity therefore gives the equivalence
\begin{equation}\label{cov-eq:schur-reduction}
  \Cov(Z)
  \preceq\eta\Diag(D_E,D_V)
  \quad\Longleftrightarrow\quad
  \Cov(X)\preceq\eta A^{-1}.
\end{equation}
Set $\widetilde\theta:=AD_E\theta$. This is an invertible change of
coefficients, and
\begin{align}
  \widetilde\theta^{\mathsf T}A^{-1}\widetilde\theta
  &=\theta^{\mathsf T}D_EAD_E\theta\notag\\
  &=\sum_{e\in E}\mu(e)\theta_e^2
    +\sum_{u\in V}
      \frac{\bigl(\sum_{e\sim u}\mu(e)\theta_e\bigr)^2}{\mu(u^0)}.
  \label{cov-eq:denominator}
\end{align}
Moreover, using $I_GX=\mathbf1-J$ and discarding an additive constant
inside the variance,
\begin{align*}
  \widetilde\theta^{\mathsf T}
    \Cov(X)\widetilde\theta
  &=\Var_\mu\!\left[
    \sum_{e\in E}\theta_eX_e
    -\sum_{u\in V}
      \frac{\sum_{e\sim u}\mu(e)\theta_e}{\mu(u^0)}J_u
  \right]\\
  &=\Var_\mu[\mathrm I+\mathrm{II}].
\end{align*}
Apply
\[
  \sqrt{\Var_\mu[\mathrm I+\mathrm{II}]}
  \le\sqrt{\Var_\mu[\mathrm I]}
    +\sqrt{\Var_\mu[\mathrm{II}]}
\]
and use \eqref{cov-eq:first-target}--\eqref{cov-eq:denominator}. This bounds every
Rayleigh quotient in \eqref{cov-eq:schur-reduction} by the claimed constant.
\end{proof}

\paragraph{Controlling the centered dimer statistic.}

The centered variables $Y_e$ have zero covariance on disjoint edges,
so only incident pairs remain. We control these pairs by conditional
monomer covariance, avoiding a degree-dependent bound on their number.
The following consequence of \cite[Theorem 13]{ConcurrentHolant26}
gives the required bound; its weighted form is also used in
Section~\ref{sec:preconditioned-js}.

\begin{lemma}[Monomer covariance and absolute row sums]
\label{thm:monomer-input}
This lemma allows arbitrary positive edge activities $\lambda_e$ and
positive monomer weights $z_v$: let
\[
 \mu(M)\propto\prod_{e\in M}\lambda_e
                  \prod_{v:\,J_v(M)=1}z_v
 \qquad(M\text{ a matching of }G).
\]
After every feasible conditioning $\tau$ on monomer coordinates,
\begin{equation}\label{eq:monomer-input}
 \begin{aligned}
 \Cov_{\mu^\tau}(J)&\preceq 2\Diag((\mu^\tau(v^0))_{v\in V}),\\
 \sum_{w\ne v}|\Cov_{\mu^\tau}(J_v,J_w)|
   &\le \mu^\tau(v^0)(1-\mu^\tau(v^0)),\\
 \sum_{w\in V}|\Cov_{\mu^\tau}(J_v,J_w)|
   &\le 2\mu^\tau(v^0)(1-\mu^\tau(v^0))\le2\mu^\tau(v^0).
 \end{aligned}
\end{equation}
The absolute row-sum assertion is stated separately from the spectral
assertion; it is not inferred from the latter.
\end{lemma}

\begin{proof}
Use the Holant signatures $f_v=(z_v,1,0,\ldots)$ and the given
edge activities. Monomer pinnings are exact-degree constraints.
By \cite[Theorem 13]{ConcurrentHolant26}, based on the
Chen--Gu coupling, the occupied-degree vector $D=\one-J$ has
degree coupling independence with constant two, also under these
pinnings. For $0<\mu^\tau(v^0)<1$, couple the laws conditioned on $J_v=1$
and $J_v=0$. The root contributes one discrepancy, leaving expected
Hamming discrepancy at most one outside $v$. Hence
\[
\begin{aligned}
 \sum_{w\ne v}|\Cov_{\mu^\tau}(J_v,J_w)|
 &=\mu^\tau(v^0)(1-\mu^\tau(v^0))\\
 &\quad{}\times\sum_{w\ne v}
 \left|\E_{\mu^\tau}[J_w\mid J_v=1]
       -\E_{\mu^\tau}[J_w\mid J_v=0]\right|\\
 &\le \mu^\tau(v^0)(1-\mu^\tau(v^0)).
\end{aligned}
\]
For deterministic $J_v$ the row is zero. Adding the diagonal gives
the full row bound, and symmetry with $2|a_va_w|\le a_v^2+a_w^2$
gives $\Var(\sum_v a_vJ_v)\le2\sum_v\mu^\tau(v^0)(1-\mu^\tau(v^0))a_v^2$.
This proves the spectral assertion as well.
\end{proof}

The weighted extension above concerns only monomer covariance.
The rest of this proof uses common edge activity $\lambda$ and unit
monomer weights; no nonuniform-activity version of
\cref{thm:joint-covariance} or of the algorithms is asserted here.

\begin{lemma}[Centered-dimer variance bound]\label{cov-lem:first-variance}
For every $\theta\in\mathbb R^E$,
\[
  \Var_\mu\!\left[\sum_{e\in E}\theta_eY_e\right]
  \le(1+3\lambda)\sum_{e\in E}\mu(e)\theta_e^2
  +\sum_{u\in V}
    \frac{\bigl(\sum_{e\sim u}\mu(e)\theta_e\bigr)^2}{\mu(u^0)}.
\]
Thus \eqref{cov-eq:first-target} holds with $\eta_1(\lambda)=1+3\lambda$.
\end{lemma}

\begin{proof}
The harmonic surrogate of \cite[Definition 17]{ConcurrentHolant26}
is $K_{uv}=\lambda J_uJ_v$ for matchings. Its conditional-mean
identity (Lemma 18 there), together with insertion and deletion, gives
\[
  \mathbb E_\mu Y_e=0,\qquad
  \mathbb E_\mu Y_e^2=(1+\lambda)\mu(e),\qquad
  \mathbb E_\mu[Y_eY_f]=0\quad(e\cap f=\varnothing).
\]
For distinct incident edges,
$\E[Y_{uv}Y_{uw}]=\lambda^2\mu(u^0v^0w^0)$,
since all terms containing an occupied edge vanish.
Put $S_u=\sum_{v\sim u}\theta_{uv}J_v$. By \eqref{eq:monomer-input}
and the insertion identity,
\[
 \lambda^2\mu(u^0)\E_{\mu^{u^0}}S_u^2
 \le2\lambda\sum_{e\sim u}\mu(e)\theta_e^2
   +\frac{\bigl(\sum_{e\sim u}\mu(e)\theta_e\bigr)^2}{\mu(u^0)}.
\]
This is the star-moment estimate from the matching-specific proof of
Lemma 7 in \cite[Section 4]{ConcurrentHolant26}, written in terms of
the second moment and retaining coefficient one on the squared mean.
Grouping incident pairs by their unique common endpoint gives
\begin{align*}
  \Var_\mu\!\left[\sum_e\theta_eY_e\right]
  &=(1-\lambda)\sum_e\mu(e)\theta_e^2
    +\lambda^2\sum_u\mu(u^0)
      \E_{\mu^{u^0}}S_u^2\\
  &\le(1+3\lambda)\sum_e\mu(e)\theta_e^2
    +\sum_u
      \frac{\bigl(\sum_{e\sim u}\mu(e)\theta_e\bigr)^2}{\mu(u^0)}.
\end{align*}
Each edge is counted at both endpoints, explaining both the diagonal
correction $1-\lambda$ and the final coefficient $1+3\lambda$.
\end{proof}

\paragraph{Controlling the corrected monomer statistic.}

The vertex terms in $\mathrm{II}$ center each corrected star sum
conditional on its center being unmatched. We preserve this cancellation
when revealing a random subset of monomer indicators.

We use the following Poincar\'e inequality for a block heat-bath chain
on the monomer indicators.

\begin{lemma}[Monomer block heat bath {\cite[Theorem 14]{ConcurrentHolant26}}]
\label{cov-lem:monomer-block-gap}
Let $P$ be the following kernel on monomer configurations: include each
vertex independently in $R$ with probability $1/3$, independently of
the current configuration; retain $J_R$ and resample the other monomer
indicators jointly from their conditional law given $J_R$.
The kernel is reversible for the law of $J$ under $\mu$, and for every
real function $F$ of $J$,
\begin{equation}\label{cov-eq:monomer-block-poincare}
 \cE_P(F,F)
 =\E_R\E_\mu\Var_\mu(F(J)\mid J_R)
 \ge\tfrac13\Var_\mu(F(J)).
\end{equation}
In particular, $\gap(P)\ge1/3$ whenever the monomer law is nontrivial.
\end{lemma}
\begin{proof}
For each fixed $R$, conditional resampling is reversible and its
Dirichlet form is $\E_\mu\Var_\mu(F(J)\mid J_R)$. Averaging over $R$
gives the displayed identity. Matchings have capacity-one signatures
$f_v=(1,1,0,\ldots)$, so \cite[Theorem 13]{ConcurrentHolant26}
verifies the hypotheses of \cite[Theorem 14]{ConcurrentHolant26}
for their occupied degrees $D=\mathbf1-J$. Conditioning on $D_R$
and on $J_R$ is the same. Equation (3.2) there, with $r=1/3$,
therefore gives the inequality.
\end{proof}

We next bound this Dirichlet form for $F=\mathrm{II}$. The calculation
follows \cite[proof of Corollary 16]{ConcurrentHolant26}, retaining
both the edge and vertex terms of $\mathrm{II}$.

\begin{lemma}[Corrected-monomer variance bound]\label{cov-lem:second-variance}
For every $\theta\in\mathbb R^E$,
\[
  \Var_\mu[\mathrm{II}]
  \le36\lambda\sum_{e\in E}\mu(e)\theta_e^2.
\]
Thus \eqref{cov-eq:second-target} holds with $\eta_2(\lambda)=36\lambda$.
\end{lemma}
\begin{proof}
Put $F=\mathrm{II}$, $K_{uv}=\lambda J_uJ_v$ and
$Q=\sum_e\mu(e)\theta_e^2$. By
\eqref{cov-eq:monomer-block-poincare}, it suffices to prove
$\cE_P(F,F)\le12\lambda Q$. Write
\[
 h_u=\frac{\sum_{e\sim u}\mu(e)\theta_e}{\mu(u^0)}
 =\lambda\sum_{v\sim u}\mu^{u^0}(v^0)\theta_{uv}.
\]
Then
\[
 A_u:=\sum_{e\sim u}\theta_eK_e-h_uJ_u
 =\lambda J_u\sum_{v\sim u}\theta_{uv}
       \bigl(J_v-\mu^{u^0}(v^0)\bigr).
\]
Thus $A_u=0$ when $J_u=0$ and $\E[A_u\mid J_u=1]=0$.
The conditional monomer covariance bound \eqref{eq:monomer-input}
and the insertion identity give
\begin{equation}\label{cov-eq:corrected-star-moments}
\begin{aligned}
 \E A_u^2
 &=\lambda^2\mu(u^0)\Var_{\mu^{u^0}}
       \left(\sum_{v\sim u}\theta_{uv}J_v\right)
 \le2\lambda\sum_{e\sim u}\mu(e)\theta_e^2,\\
 \E K_e^2&=\lambda\mu(e).
\end{aligned}
\end{equation}

For the random set $R$ in the block update, write
$B_u=\one_{\{u\in R\}}$ and define
\[
 F_R=9\sum_{uv\in E}\theta_{uv}K_{uv}B_uB_v
       -3\sum_uh_uJ_uB_u.
\]
For fixed $R$, $F_R$ is determined by $J_R$. The conditional mean
minimizes mean-square error, so
\[
 \cE_P(F,F)=\E_R\E_\mu\Var(F\mid J_R)
 \le\E_R\E_\mu(F-F_R)^2.
\]
To evaluate the right side, put $\xi_u=3B_u-1$.
These independent variables have mean zero and second moment two, and
\[
 F_R-F=\sum_u\xi_uA_u
       +\sum_{uv\in E}\theta_{uv}K_{uv}\xi_u\xi_v.
\]
Conditional on the matching, the distinct linear and quadratic
monomials are orthogonal: simplicity ensures that distinct edges
index distinct pairs of vertices. Hence
\[
 \E_R\E_\mu(F_R-F)^2
 =2\sum_u\E A_u^2+4\sum_e\theta_e^2\E K_e^2
 \le12\lambda Q,
\]
by \eqref{cov-eq:corrected-star-moments}, counting each edge at both
endpoints. Hence $\cE_P(F,F)\le12\lambda Q$, and
\eqref{cov-eq:monomer-block-poincare} gives
$\Var F\le3\cE_P(F,F)\le36\lambda Q$.
\end{proof}
\begin{proof}[Proof of Theorem~\ref{thm:joint-covariance}]
Lemmas~\ref{cov-lem:first-variance} and~\ref{cov-lem:second-variance} establish
the two estimates in Lemma~\ref{cov-lem:reduction} with
\[
  \eta_1(\lambda)=1+3\lambda,\qquad
  \eta_2(\lambda)=36\lambda.
\]
The covariance bound follows from that reduction. Its dimer principal block
gives \eqref{eq:covariance}. A feasible edge pinning leaves a matching
law on a vertex- and edge-deleted simple graph, together with deterministic
coordinates. Applying the same proof to that residual graph gives the
pinned statement as well.
\end{proof}

\section{Extensions to Jerrum--Sinclair dynamics and \texorpdfstring{$b$}{b}-matchings}
\label{app:local-matching-chains}

We record two consequences of the mixing analysis, neither needed in the
main proofs: a comparison for classical Jerrum--Sinclair dynamics and an
application of the general criterion to capacity-constrained matchings.

\subsection{The Jerrum--Sinclair chain}
Insertion and deletion moves alone allow us to transfer the Glauber
bounds to the classical Jerrum--Sinclair chain
\cite{JerrumSinclair,CFJMYZ25}, without the learned parameters of
Section~\ref{sec:preconditioned-js}.
We use its $1/2$-lazy version,
denoted by $P_{\rm JS}$. With probability $1/2$ it stays put;
otherwise it chooses an edge $e=uv$ uniformly. If $e$ is occupied,
the proposal deletes it. If both endpoints are unmatched, the proposal
inserts it. If exactly one endpoint is matched, the proposal replaces
that endpoint's matching edge by $e$. In the remaining case it stays
put. A proposed $M'$ is accepted with probability
$\min\{1,\lambda^{|M'|-|M|}\}$.

\begin{corollary}\label{cor:original-js}
For every $\lambda>0$, the $1/2$-lazy Jerrum--Sinclair chain defined above
on an $n$-vertex simple graph with $m\ge1$ edges satisfies, for
$0<\varepsilon\le1/2$,
\[
 \alpha_{\rm LS}(P_{\rm JS})^{-1}=O_\lambda(m\log n),\qquad
 t_{\rm mix}(P_{\rm JS},\varepsilon)
 =O_\lambda\!\left(m[\log^2n+\log(1/\varepsilon)]\right).
\]
Each step takes $O(1)$ work after $O(m+n)$ initialization.
\end{corollary}
\begin{proof}
The proposal is symmetric: an exchange is reversed by proposing the edge
it removed. The Metropolis acceptance rule therefore makes the chain
reversible for $\mu_{G,\lambda}$; insertion and deletion ensure
irreducibility, and laziness ensures nonnegative spectrum.
For every $e\in M$, the stationary transition weight is
\[
 \mu(M)P_{\rm JS}(M,M-e)
 =\frac{\mu(M)}{2m\max\{1,\lambda\}}.
\]
Insertion and deletion therefore give
\[
 \cE_{P_{\rm JS}}(f,f)
 \ge\frac{1+\lambda}{2\max\{1,\lambda\}}\cE_{\mathrm{GD}}(f,f)
 \ge\tfrac12\cE_{\mathrm{GD}}(f,f),
\]
since exchange moves add nonnegative terms. The comparison transfers both
the LSI in \eqref{eq:matching-LS} and the $O_\lambda(m)$ relaxation-time
bound of \cref{prop:matching-gap}, due to
\cite[Corollary 5]{ConcurrentHolant26}.
Together with \eqref{eq:matching-min-mass}, the LSI-to-mixing estimate
\eqref{green-eq:ls-mixing} gives the stated accuracy dependence.
Storing each vertex's matched neighbor (or an unmatched flag) makes every
move take $O(1)$ work. Thus TV-$\varepsilon$ sampling takes
$O_\lambda(n+m[\log^2n+\log(1/\varepsilon)])$ work.
\end{proof}

\subsection{\texorpdfstring{$b$}{b}-matchings}
To distinguish capacities from the marginal parameter $b$, let $B\ge1$
be an integer and write $0\le B_v\le B$ for integer vertex capacities.
On an $n$-vertex simple graph $G=(V,E)$ with $m\ge1$ edges, define
\[
 \Omega=\{M\subseteq E:d_M(v)\le B_v\ \text{for all }v\},
 \qquad \mu_\lambda(M)\propto\lambda^{|M|}.
\]
The following extension applies \cref{thm:intro-boolean} using the
degree-independent covariance estimates of
\cite[proofs of Theorem 4 and Corollary 5]{ConcurrentHolant26}.

\begin{corollary}\label{cor:b-matching}
For every integer $B\ge1$ and $\lambda>0$, single-edge Glauber dynamics for
$\mu_\lambda$ satisfies, for $0<\varepsilon\le1/2$,
\[
 \alpha_{\rm LS}^{-1}=O_{B,\lambda}(m\log n),\qquad
 t_{\rm mix}(\varepsilon)
 =O_{B,\lambda}\!\left(m[\log^2 n+\log(1/\varepsilon)]\right).
\]
At $\lambda=1$, the respective bounds are
$O(Bm\log n)$ and
$O(Bm[\log^2 n+\log(1/\varepsilon)])$, with absolute constants.
The bounds apply to instances with at least two feasible configurations;
a one-state instance has mixing time zero. Each update takes $O(1)$ work
after $O(m+n)$ initialization.
\end{corollary}
\begin{proof}
Positive pinning deletes the pinned edges and reduces the capacities
by their occupied degrees, preserving the form of the law. Edges forced
absent are omitted from covariance matrices but remain among the $m$
update coordinates of the original chain.

\emph{Spectral stability.}
Write $X_e=\mathbf1_{\{e\in M\}}$; the decreasing field has activity
$s=(1-t)\lambda$. The low-activity estimate
\cite[Theorem 11, equation (2.3)]{ConcurrentHolant26} states, for every positive pinning $\tau$,
\[
 \Cov_{\mu_s^\tau}(X)\preceq
 \frac1{1-s}\operatorname{Diag}(\E_{\mu_s^\tau}X),\qquad 0<s<1.
\]
Put $a=\min\{\lambda,1/2\}$. The proof of Theorem 4 of the same reference, with
insertion odds between $a$ and $\lambda$, bounds these covariance matrices
by $A\operatorname{Diag}(\E X)$ uniformly for $a\le s\le\lambda$,
where $A=O_{B,\lambda}(1)$ may be enlarged to satisfy $A\ge2$.
Consequently, an admissible nonincreasing rate is
\[
 C(t)=
 \begin{cases}
 A,&(1-t)\lambda>a,\\
 [1-(1-t)\lambda]^{-1},&(1-t)\lambda\le a,
 \end{cases}
 \qquad
 I=(A-1)\log\frac\lambda a-\log(1-a)<\infty.
\]
For $\lambda=1$, \cite[proof of Corollary 5]{ConcurrentHolant26}
gives the uniform covariance coefficient $A=1804B$ for all $0<s\le1$.
We can then take $C(t)=\min\{A,1/t\}$, with $C(0)=A$, obtaining
\[
 I=-(A-1)\log(1-A^{-1})+\log A\le1+\log A,
 \qquad e^I\le1804eB.
\]

\emph{Marginal bounds.}
Every legal insertion has weight ratio $\lambda$, so take
$\kappa=\max\{1,\lambda\}$. To bound occupied marginals, fix a feasible
edge $e=uv$ in any residual instance. Map each configuration not containing $e$ to
one containing it by deleting one incident occupied edge at each
saturated endpoint, chosen in a fixed order, and then inserting $e$.
The two deleted edges, if both present, are distinct because the graph
is simple and $e$ was absent. An image $M'$ and the two deletion records
uniquely determine the preimage, whose weight is $\lambda^{|M'|-1+k}$
if $k$ edges were deleted. Summing over possible records bounds the
total preimage weight by
\[
 \lambda^{|M'|-1}(1+\lambda d_G(u))(1+\lambda d_G(v)).
\]
Adding the configurations that already contain $e$ and summing gives
\[
 \mu_\lambda^\tau(e\in M)
 \ge\frac{\lambda}
 {\lambda+(1+\lambda d_G(u))(1+\lambda d_G(v))}
 \ge\frac{\lambda}{\lambda+(1+\lambda\Delta)^2}.
\]
Here $\Delta$ is the original graph's maximum degree.
Apply \cref{thm:intro-boolean} with $N=m$ and the last expression as $b$.
Since $\Delta\le n-1$ and $m\le\binom n2$,
$\log(m\kappa/b)=O_\lambda(\log n)$.
Also $\log(1/\mu_{\min})\le m(\log2+|\log\lambda|)$, since there are at
most $2^m$ states and the ratio of any two weights is at most
$\max\{\lambda,\lambda^{-1}\}^m$. The LSI and mixing bounds follow.
Storing the edge occupations and each vertex's occupied degree makes
every Glauber update take $O(1)$ work. Running for the stated mixing-time
bound gives the same sampling-time bound, plus $O(m+n)$ initialization
and output.
\end{proof}
\end{document}